\documentclass[journal]{IEEEtran}

\usepackage{url}
\usepackage{bbm}
\usepackage{cite}
\usepackage{array}
\usepackage{xcolor}
\usepackage{comment}
\usepackage{enumitem}
\usepackage{makecell}
\usepackage{booktabs}
\usepackage{colortbl}
\usepackage{verbatim}
\usepackage{graphicx}
\usepackage{textcomp}
\usepackage{stfloats}
\usepackage{flushend}
\usepackage{algorithm}
\usepackage{algorithmic}
\usepackage{amsmath, amssymb, amsfonts, mathtools}
\usepackage[protrusion=false, expansion=true, spacing=true]{microtype}
\usepackage[caption=false,font=normalsize,labelfont=sf,textfont=sf]{subfig}
\usepackage[colorlinks, linkcolor=blue, anchorcolor=black, citecolor=blue]{hyperref}

\newtheorem{lemma}{Lemma}

\newtheorem{theorem}{Theorem}

\setlist[itemize]{left=1pt}
\setlist[enumerate]{left=1pt}

\begin{document}

\title{
	Precoding Sequence Design for MIMO Sensing with Scatterers Based on Prior Information
}

\author{
	Yiming~Liu,~\IEEEmembership{Graduate~Student~Member,~IEEE,}
	and
	Wei~Yu,~\IEEEmembership{Fellow,~IEEE}
		
	\thanks{
		Manuscript submitted for possible publication on June 13, 2026. 
		This work was supported by NSERC via the Canada Research Chairs program and Discovery Grant.
		The materials in this paper has been presented in part at the IEEE International Symposium on Information Theory (ISIT), June 2026 \cite{LiuISIT}.
		\textit{(Corresponding author: Wei~Yu.)}
		
		The authors are with The Edward S. Rogers Sr. Department of Electrical and Computer Engineering, University of Toronto, Toronto, Ontario M5S 3G4, Canada (e-mails: eceym.liu@mail.utoronto.ca; weiyu@ece.utoronto.ca).
	}
		
}




\maketitle

\begin{abstract}

The presence of interfering scatterers fundamentally changes the design principle for MIMO sensing.
Unlike the target-only case,
where MIMO sensing sequence design reduces to optimizing the transmit sample covariance,
this paper shows that 
scatterer-induced signal-dependent interference makes the Bayesian Fisher information depend on the full temporal precoding sequence.
Consequently,
for the MIMO sensing problem with scatterers using the Bayesian Cramér-Rao lower bound (BCRLB) as the objective,
the entire sensing sequence must be designed explicitly,
instead of just the precoding matrix.
This paper considers such a precoding sequence design problem under hardware constraint
for MIMO sensing for estimating the azimuth angles of multiple targets
based on the prior information of both the targets and the scatterers.
We formulate a worst-case BCRLB minimization across multiple target angles, 
yielding a max-min fractional program under constant-modulus or constant-norm hardware constraints.
We further develop a constant-norm linear transform that converts the ratio objectives into linear forms,
leading to an iterative algorithm with closed-form precoder updates.
The framework extends to joint precoder-combiner design and multi-stage sensing with adaptive prior refinement.
Numerical results demonstrate the effectiveness and the efficiency of the proposed algorithm,
revealing sweeping-like beampatterns that illuminate target angular regions
while suppressing interference from the scatterers. 

\end{abstract}

\begin{IEEEkeywords}
	
Bayesian inference, Cramér-Rao lower bound, MIMO sensing, precoding, fractional programming, max-min optimization, constant-norm linear transform, multi-stage sensing.

\end{IEEEkeywords}

\section{Introduction}
\label{Section I}

\IEEEPARstart{M}{ulti}-input multiple-output (MIMO) technology,
owing to its capability to exploit spatial dimensions to enhance spectral efficiency and mitigate interference,
has long served as a cornerstone of wireless communication networks \cite{6736761}.
In recent years, emerging applications, e.g., smart cities, industrial automation, and autonomous driving, 
have created an increasing demand for sensing functionalities in beyond-fifth-generation (B5G) and sixth-generation (6G) wireless networks.
The trend of wireless communication systems towards higher frequency bands and larger antenna arrays
provides an exciting opportunity to implement sensing by utilizing existing MIMO infrastructure, i.e., MIMO sensing.
Moreover, MIMO sensing can be integrated with communication functionalities, i.e., integrated sensing and
communications (ISAC),
enabling future wireless networks to provide intelligent
sensing services and environment-aware applications \cite{9737357}.
In this paper, we investigate how to effectively leverage MIMO for sensing.

MIMO sensing has attracted significant attention in the recent literature,
but most existing works on the design of MIMO sensing strategies assume a target-only environment,
in which the received echoes are assumed to be associated solely with the targets of interest,
while scatterers are either neglected or treated as additional targets to be estimated.
The main objective of this paper is to show that 
by taking into account the potential availability of \emph{prior information} on the targets and the scatterers, 
it is possible to take advantage of the spatial degrees-of-freedom (DoFs) offered by MIMO arrays 
to spatially separate targets from scatterers.
Toward this end,
this paper adopts a Bayesian approach
in which the difference between the prior distributions of the target and scatterer parameters
is utilized to design spatial precoders to enhance target-reflected signals
while suppressing scatterer-induced interference.
This allows us to account for the scatterers without estimating them as additional targets,
which would have resulted in an excessive increase in the dimensionality of the parameter space.

Specifically, this paper employs the Bayesian Cramér-Rao lower bound (BCRLB) as the sensing performance metric
and design MIMO precoders to minimize the worst-case BCRLB across multiple targets,
while accounting for the statistics of the scatterers.
We show that the presence of interfering scatterers fundamentally changes the design principle for MIMO sensing.
In the target-only scenarios commonly considered in existing works,
the MIMO precoder design can be reduced to the design of the sample covariance of the transmit signal.
In contrast, 
in this paper we show that because scatterer-induced interference is \emph{signal-dependent},
one must explicitly optimize the MIMO precoding sequence over the entire sensing interval,
instead of just optimizing its sample covariance.

This paper develops an optimization paradigm for the MIMO sensing sequence design
problem
while accounting for practical hardware constraints.
At the transmitter side,
many existing studies assume a fully digital architecture,
which requires one dedicated radio frequency (RF) chain per antenna,
resulting in high hardware cost and power consumption \cite{1367565}.
An alternative is the single-RF-chain architecture, 
where the transmit precoder is implemented via phase shifters.
In this paper, we consider both the single-RF-chain architecture
and the architecture where both the amplitude and phase can be controlled.
We show that these hardware constraints can be effectively dealt with 
using a novel \emph{constant-norm linear transform technique},
which gives rise to a low-complexity iterative algorithm that admits closed-form precoder updates in each iteration.

At the receiver side, this paper considers both the fully digital and the 
single RF-chain architectures.
In the latter case, 
the received signals across multiple antennas need to be combined before further processing.
Consequently, the sensing performance depends jointly on both the precoders and the combiners.
This motivates the joint design of precoders and combiners treated in this paper.

The proposed sensing scheme can be extended into a multi-stage setting, where the posterior of the target angles is updated based on the echoes received within each sensing stage and is used as the prior for the next stage.
Accordingly, the precoders are designed stage by stage based on the refined prior, leading to improved estimation accuracy in the presence of interfering scatterers.
This paper analyzes such a sequential
sensing strategy by quantifying the impact of stage length on multi-stage sensing performance under a fixed symbol-period budget,
thus providing guidelines for designing multiple-stage sensing schemes.

\subsection{Related Works}

There are numerous prior works investigating MIMO sensing with or without integration into communications.
Most of them consider only the targets of interest while neglecting interfering scatterers,
which are often present in realistic environments.
In such scenarios,
the transmission design for MIMO sensing often reduces to optimizing a fixed transmit sample covariance.
In earlier MIMO sensing works,
the transmit sample covariance matrix is designed based on various sensing objectives,
such as matching desired beampatterns \cite{4276989, 5765721},
and minimizing the CRLB for estimation \cite{4359542, 10138058}.
In recent ISAC studies,
the core focus is still on sensing-oriented transmission design,
but under extra communication constraints.
For example, the work \cite{8386661} designs the beampatterns for an ISAC system in order to achieve a trade-off between two functionalities.
To facilitate the MIMO radar in utilizing its full DoF, 
the authors in \cite{9124713} jointly design the individual radar and communication transmit covariance matrices to optimize both the beampatterns and the signal-to-interference-and-noise ratio (SINR).
In \cite{9652071}, the design of the transmit covariance matrices is formulated to minimize the sensing CRLB under communication constraints.
The transmit covariance can also be optimized for CRLB-based sensing under SINR constraints via uplink-downlink duality \cite{2509.13661}.
The transmit covariance has also been studied from an information-theoretic perspective,
where the CRLB-rate region is used to characterize the fundamental sensing-communication trade-off under Gaussian ISAC channels \cite{10147248}.
In addition, due to the capability of reconfigurable intelligent surfaces (RISs) to control the propagation environment,
several works have investigated RIS-assisted MIMO sensing \cite{10440056, 9454375, 9732186, 9769997, 10364735}.

Beyond the target-only settings discussed above,
some prior works have also considered interference in the sensing model.
For example, the works \cite{6748061, 7953658} incorporate interference or clutter into the sensing model,
but their designs still remain at the level of a fixed transmit covariance.
References \cite{6649991, 8141978} consider sensing with interference and design a sequence of waveforms to maximize the sensing SINR,
with the formulation restricted to a single-target scenario.
Interference from external communication systems is considered in \cite{8892631}.

Most of the aforementioned works,
including both target-only and interference-aware formulations,
are conducted with the assumption of fully digital transceivers,
which may lead to high hardware costs and power consumption.
Some studies use low-resolution digital-to-analog converters (DACs) or analog-to-digital converters (ADCs) to solve this issue \cite{Liu2506:MIMO}.
For example, \cite{9266762} considers only the one-bit ADCs at receivers and uses the SNR as a metric for sensing performance.
In \cite{9399801}, the authors consider only the one-bit DACs at transmitters and adopt the CRLB as the sensing metric.
Another approach is to employ fewer RF chains at transceivers.
For example, \cite{10901673} employs a hybrid analog-digital transmitter to reduce the cost of ISAC systems.
This paper considers phase-only precoding/combining at transceivers to reduce hardware cost and power consumption.

It is also worth noting that most of the aforementioned works optimize sensing transmissions based on heuristic metrics,
such as the sensing SINR and beampattern matching.
Another widely used metric is the CRLB,
which is more directly connected to the sensing performance since it provides a lower bound on the mean-squared error (MSE) among all unbiased estimators.
However, it has an inherent limitation that it depends on the true values of the parameters to be estimated (such as the azimuth angles), which are unknown in practice.
A more practical framework is to assume a prior distribution for the parameters and to employ the BCRLB (e.g., \cite{6541985}) as a more suitable sensing metric, as done in \cite{10901183, 10584278}.
In particular, \cite{10901183} adopts the BCRLB as the metric to design the probing signals for a MIMO sensing system with fully digital transceivers.
The work \cite{10584278} uses the BCRLB as the sensing metric for an ISAC system.
However, most of these prior works do not account for scatterers. In addition, many of the aforementioned studies rely on computationally intensive optimization methods,
e.g., semidefinite programming,
which have limited practicality for large-scale MIMO especially in real-time applications.

\subsection{Main Contributions}

The main contributions of this paper are as follows:

\begin{itemize}
	
	\item
	This paper adopts a Bayesian approach for target sensing in the presence of scatterers.
	We show that spatial beamforming can be used to exploit the difference in the prior distributions of the target and the scatterer parameters to distinguish the targets from the scatterers.
	
	\item
	This paper models the scatterers as signal-dependent Gaussian interference
	and reveals that the presence of interfering scatterers makes the Bayesian Fisher information for target-sensing depend on the entire temporal precoding sequence rather than just the transmit sample covariance.
	This necessitates precoding sequence design for MIMO sensing in the presence of scatterers.
	
	\item
	This paper formulates a BCRLB-based precoding sequence design problem for estimating the azimuth angles of multiple targets in the presence of scatterers.
	We show that when the nuisance pathloss coefficients are modelled as independent zero-mean Gaussian random variables,
	it leads to tractable BCRLB expressions.
	In this case, the problem of minimizing the worst-case BCRLB across the targets can be cast as a max-min fractional program.
	
	\item
	This paper considers both the single-RF-chain transceiver architecture where the precoders are subject to constant-modulus constraint, 
	and the architecture where both amplitude and phase can be controlled and the precoder must satisfy a constant-norm constraint.
	Both of these types of constraints can be efficiently dealt with 
	by using a linear transform technique
	that converts the ratio objective into a linear form,
	which leads to an iterative algorithm with closed-form precoder update in each iteration.

	\item 
	This paper further extends the proposed framework to joint precoder-combiner design 
	and to multi-stage sensing with adaptive prior refinement.
	
	\item 
	Numerical results show nontrivial sequential sweeping-like beampatterns that allow target signals to be enhanced,
	while suppressing the interference due to scatterers.
	For the multi-stage sensing scenario, 
	the results also reveal a trade-off between the number of sensing stages and the number of sensing symbols per stage.
	
\end{itemize}

\subsection{Notations}
 
This paper uses lowercase letters, bold lowercase letters, and bold uppercase letters to denote scalars, vectors, and matrices, respectively.
The transpose, Hermitian transpose, inverse, and trace of matrices are represented by $\left( \cdot \right)^{\mathsf{T}}$, $\left( \cdot \right)^{\mathsf{H}}$, $\left( \cdot \right)^{-1}$, and $\mathsf{Tr}\left( \cdot \right)$, respectively. 
The operators $\mathfrak{Re} \left\lbrace \cdot \right\rbrace$ and $\mathfrak{Im} \left\lbrace \cdot \right\rbrace$ are the real and imaginary parts of a complex number.
We adopt $\mathbf{I}_N$ to denote an $N \times N$ identity matrix.
The expectation is denoted by $\mathbb{E} \left[ \cdot \right]$.
We adopt the operator $\mathsf{vec} \left( \cdot \right)$ to stack the columns of a matrix into a column vector and adopt $\mathsf{vec}^{-1} \left( \cdot \right)$ to denote its inverse operation.
We use $\angle (\cdot)$ to denote the phase of a complex number 
and use $\mathcal{CN} ( \mathbf{x} \mathop{|} \boldsymbol{\mu}, \mathbf{\Sigma} )$ to denote the complex Gaussian distribution of $\mathbf{x}$ with mean $\boldsymbol{\mu}$ and covariance $\mathbf{\Sigma}$.

\section{System Model} 
\label{Section 02}

This paper investigates a monostatic MIMO sensing system equipped with colocated $N_T$ transmit antennas and $N_R$ receive antennas.
Both the transmit and receive antennas are arranged as uniform linear arrays for simplicity, as illustrated in Fig.~\ref{Figure_01}.
The MIMO sensing scenario consists of $M$ targets of interest and $K$ scatterers.
The precoded signals from the transmit antennas are reflected by both the targets and the scatterers.
The sensing task is to estimate the azimuth angles of the targets of interest, based on the received echoes.

\begin{figure*}[!t]
	\centering
	\includegraphics[width=1\linewidth]{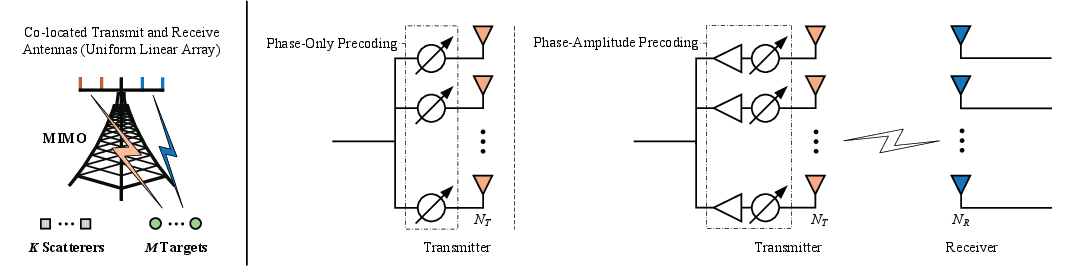}
	\caption{Monostatic MIMO sensing with phase-only or phase-amplitude transmitter and fully digital receiver.}
	\label{Figure_01}
\end{figure*}

We adopt a Bayesian approach and use the knowledge of the prior probability distributions of the azimuth angles of the targets and the scatterers to separate them spatially.
However, we do not estimate or track the azimuth angles of the scatterers in order to prevent excessive dimensionality increase in the parameter space.
In this scenario, the signals reflected from the scatterers become interference in the estimation process for the target angles.
Moreover, we also do not estimate or track the pathloss parameters of the targets and the scatterers and assume that they are nuisance random variables with zero-mean and known variances.
This zero-mean assumption significantly simplifies the mathematical derivation;
such assumption is often made in the radar and mmWave literature \cite{skolnik2001radar, 6847111}, (especially when the physical characteristics of the targets are completely unknown, at least in the initial sensing stage). 
The design task is to find a sequence of precoded signals across the MIMO arrays and across the sensing interval to optimize the performance of estimating the target angles.

To reduce hardware cost and power consumption,
we begin by assuming that the transmitter is equipped with a single RF chain,
while the receiver is fully digital\footnote{The extension to a single-RF-chain receiver is investigated in Section~\ref{Section_06}.}.
The sequence of transmit precoders across $N_T$ antennas and $N$ symbol periods can be expressed as
\begin{equation}
	\mathbf{x}
	\triangleq
	\mathsf{vec} \left( \left[ \mathbf{x}_1, \mathbf{x}_2, \ldots, \mathbf{x}_N \right] \right),
\end{equation}
where $\mathbf{x}_n \in \mathbb{C}^{N_T \times 1}$ is the precoder for the $n$-th symbol period.
For the single-RF-chain transmitter, precoding is implemented using phase shifters with 
independent phase control at each antenna.
This is known as phase-only precoding, i.e.,
each entry $x_i$ of $\mathbf{x}$ is constrained to have constant modulus:
\begin{equation} \label{C1}
	\text{(C1):}
	\quad
	\left| x_i \right| = 1,
	\quad
	i = 1, 2, \ldots, N_T N.
\end{equation}

Alternatively, precoding can also be implemented by coupling each phase shifter with an amplifier \cite{8750903}, \cite{10124207}, enabling independent control of both phase and amplitude at each antenna.
We refer to this type of precoding as phase-amplitude precoding.
In this paper, we assume that each precoded symbol has a constant transmit power. In this case, $\mathbf{x}_n$ is subject to a constant-norm constraint:
\begin{equation} \label{C2}
	\text{(C2):}
	\quad
	\left\| \mathbf{x}_n \right\|_2 = \sqrt{N_T}, 
	\quad
	n = 1, 2, \ldots, N.
\end{equation}

For the transmit and receive antenna arrays, we index antennas with respect to the centre of the array, defined as
\begin{align}
	t_{n}
	&\triangleq
	n - \frac{N_T+1}{2},
	\quad 
	n = 1, 2, \ldots, N_T,
	\\
	r_{n}
	&\triangleq
	n - \frac{N_R+1}{2},
	\quad
	n = 1, 2, \ldots, N_R.
\end{align}
Then,
the steering vectors for the forward path from the antennas to the target/scatterer at the azimuth angle $\theta$ and for the backward path from the target/scatterer to the antennas are 
\begin{align}
	\mathbf{h}_T(\theta)
	& \triangleq
	\left[
		e^{ j t_{1}   \tau \cos(\theta)},
		e^{ j t_{2}   \tau \cos(\theta)},
		\ldots,
		e^{ j t_{N_T} \tau \cos(\theta)}
	\right]^{\mathsf{T}},
	\\
	\mathbf{h}_R(\theta)
	& \triangleq
	\left[
		e^{ j r_{1}   \tau \cos(\theta)},
		e^{ j r_{2}   \tau \cos(\theta)},
		\ldots,
		e^{ j r_{N_R} \tau \cos(\theta)}
	\right]^{\mathsf{T}},
\end{align}
where $\tau = {2 \pi d} / {\lambda}$, $\lambda$ is the carrier wavelength,
and $d$ denotes the spacing between adjacent antennas, typically at $\lambda / 2$.
Then, the round-trip array response is given by
\begin{equation}
	\mathbf{H}(\theta)
	=
	\mathbf{h}_R(\theta)
	\mathbf{h}_T^{\mathsf{T}}(\theta).
\end{equation}

Considering the significant pathloss in high-frequency bands, we only account for first-order reflections in the received signal model.
Furthermore, the transmitted symbol is set to $1$.
This is without loss of generality since any known unit-power symbol can be absorbed into the precoders.
Then, the echo signals received across $N_R$ antennas and $N$ symbols can be stacked into a single vector $\mathbf{y} \in \mathbb{C}^{N_RN \times 1}$ as
\begin{align} \label{echo_signal}
	\mathbf{y} 
	& = 
	\mathsf{vec} \left( \left[ \mathbf{y}_1, \mathbf{y}_2, \ldots, \mathbf{y}_N \right] \right)
	= 
	\sum_{m=1}^{M} 
	\left( \alpha_m \mathbf{I}_N \otimes \mathbf{H}(\theta_m) \right)
	\left( \sqrt{p}\mathbf{x} \right)  \notag 
	\\ 
	& \hspace{28pt}
	+ \sum_{k=1}^{K}
	\left( \beta_k \mathbf{I}_N \otimes \mathbf{H}(\eta_k) \right)
	\left( \sqrt{p}\mathbf{x} \right) + \mathbf{n}  \notag  
	\\
	& \triangleq 
	\sum_{m=1}^{M} \alpha_m \sqrt{p} \mathbf{U}_m \mathbf{x}
	+ 
	\sum_{k=1}^{K} \beta_k \sqrt{p} \mathbf{\Upsilon}_k \mathbf{x}
	+ 
	\mathbf{n},
\end{align}
where we define $\mathbf{U}_m$ and $\mathbf{\Upsilon}_k$ for notational convenience.

Here, $\theta_m$ and $\eta_k$ are the azimuth angles of the $m$-th target and the $k$-th interfering scatterer, respectively,
and are assumed to be constant within $N$ symbol periods;
their prior distributions are assumed to be mutually independent and are known a priori.
In this paper, we exploit the difference in their priors to spatially separate the targets from the interfering scatterers.

The parameters $\alpha_m$ and $\beta_k$ denote the pathloss coefficients of the $m$-th target and the $k$-th interfering scatterer, respectively, which are assumed to be fixed over
the $N$ symbol periods. The parameters $\alpha_m$ and $\beta_k$ have
prior distributions of $\mathcal{CN}(0,\varsigma_m^2)$ and $\mathcal{CN}(0,\sigma_k^2)$, respectively, 
where $\varsigma_m^2$ and $\sigma_k^2$ are assumed to be known.
Finally, the transmit power per symbol period is $pN_T$, and the noise $\mathbf{n} \sim \mathcal{CN}(\mathbf{0}, \varepsilon^2 \mathbf{I}_{N_R N})$.

\section{Sensing Performance Metric}
\label{Section_03}

In this section, we introduce the sensing metric employed in the design methodology and show the impact of the interfering scatterers.
We utilize the BCRLB as the metric,
which enables exploitation of the prior information of the parameters of the targets of interest.
The BCRLB provides a lower bound on the MSE of estimators.
Specifically, this paper adopts an approach of formulating the BCRLB of the parameters of $M$ targets,
while treating the scatterers as interference.
The target parameters are represented by a $3M$-dimensional real-valued vector $\mathbf{v} \triangleq [\boldsymbol{\theta}^{\mathsf{T}}, \boldsymbol{\alpha}^{\mathsf{T}}]^{\mathsf{T}}$, where
\begin{align}
	\boldsymbol{\theta}
	& \triangleq
	\left[
		\theta_1, \theta_2, \ldots, \theta_M
	\right]^{\mathsf{T}},  
	\\
	\boldsymbol{\alpha}
	& \triangleq
	\left[
		\mathfrak{Re}\{\alpha_1\}, \ldots, \mathfrak{Re}\{\alpha_M\},
		\mathfrak{Im}\{\alpha_1\}, \ldots, \mathfrak{Im}\{\alpha_M\}
	\right]^{\mathsf{T}}.
\end{align}
The BCRLB of $\mathbf{v}$ has the following form: \cite{10584278}
\begin{equation} \label{eq:BCRLB}
	\mathsf{BCRLB} \left( \mathbf{v} \right) 
	= 
	\mathsf{Tr}
	\left\lbrace
	\left(  
	\left[
	\begin{array}{cc}
		\mathbf{\Phi}_{\mathsf{o}} & 
		\mathbf{F}_{\mathsf{o}} \\[1pt]
		\mathbf{F}_{\mathsf{o}}^{\mathsf{T}} &
		\mathbf{A}_{\mathsf{o}}
	\end{array} \right]
	+
	\left[
	\begin{array}{cc}
		\mathbf{\Phi}_{\mathsf{p}} & 
		\mathbf{F}_{\mathsf{p}} \\[1pt]
		\mathbf{F}_{\mathsf{p}}^{\mathsf{T}} &
		\mathbf{A}_{\mathsf{p}}
	\end{array} \right]
	\right)^{-1} 
	\right\rbrace,
\end{equation}
where $\mathbf{\Phi}_{\mathsf{o}}$ and $\mathbf{\Phi}_{\mathsf{p}}$ are the Bayesian Fisher information matrices (BFIMs) for $\boldsymbol{\theta}$ based on the received signals $\mathbf{y}$ and from the prior distributions, respectively;
$\mathbf{A}_{\mathsf{o}}$ and $\mathbf{A}_{\mathsf{p}}$ are the BFIMs for $\boldsymbol{\alpha}$;
$\mathbf{F}_{\mathsf{o}}$ and $\mathbf{F}_{\mathsf{p}}$ are the cross-BFIMs between $\boldsymbol{\theta}$ and $\boldsymbol{\alpha}$.

The terms $\mathbf{\Phi}_{\mathsf{p}}$, $\mathbf{A}_{\mathsf{p}}$, and $\mathbf{F}_{\mathsf{p}}$ depend solely on the prior distributions and are independent of the precoder $\mathbf{x}$.
Moreover, since the prior distributions of the target parameters are mutually independent,
$\mathbf{\Phi}_{\mathsf{p}}$ and $\mathbf{A}_{\mathsf{p}}$ are diagonal, and $\mathbf{F}_{\mathsf{p}} = \mathbf{0}$.

We use $\mathbf{\Phi}_{\mathsf{p}}$ as an example to illustrate the computation involved.
Let $q_m (\theta_m)$ denote the prior distribution of $\theta_m$.
The $m$-th diagonal entry of $\mathbf{\Phi}_{\mathsf{p}}$ is given by
\begin{equation} \label{prior_FI}
	\mathbf{\Phi}_{\mathsf{p}}^{(m,m)}
	\triangleq
	\mathop{\mathbb{E}_{\theta_m}}
	\left[
		\left(
			\frac{\mathop{\partial}}{\partial \theta_m}
			\ln q_m (\theta_m)
		\right)^2
	\right].
\end{equation}

The terms $\mathbf{\Phi}_{\mathsf{o}}$, $\mathbf{A}_{\mathsf{o}}$, and $\mathbf{F}_{\mathsf{o}}$ are induced by the received signals $\mathbf{y}$. They
depend on the precoding sequence $\mathbf{x}$ in different ways depending on whether the scatterers are taken into consideration, as explained in the subsequent sections.

The parameters of interest here are the target azimuth angles $\boldsymbol{\theta}$, 
while the pathloss coefficients $\boldsymbol{\alpha}$ are nuisance parameters.
In the proposed methodology, we first include $\boldsymbol{\alpha}$ in the BCRLB,
and then extract a lower bound for estimating $\boldsymbol{\theta}$ from the joint BCRLB expression via Schur complement at a later stage. 
This corresponds to a joint estimator for $(\boldsymbol{\alpha}, \boldsymbol{\theta})$ but with only $\boldsymbol{\theta}$ as the parameter of interest.

Moreover, the scatterers are not tracked and are treated as interference;
their effect is spatially filtered out based on the difference in the priors of the target angles $\{ \theta_m \}_{m=1}^{M}$ and the scatterer angles $\{ \eta_k \}_{k=1}^{K}$.
For tractability, the interference due to scatterers is assumed to be 
Gaussian distributed in this paper.

We now state the BCRLB optimization problems for the scenario with no interfering scatterers versus the scenario with signal-dependent interferences from the scatterers.

\subsection{Without Interfering Scatterers} 
\label{target-only}

In the target-only sensing scenarios, the echo signal model in \eqref{echo_signal} is reduced to
\begin{equation}
	\mathbf{y}
	=
	\sum_{m=1}^{M}
	\alpha_m \sqrt{p} \mathbf{U}_m \mathbf{x}
	+
	\mathbf{n},
\end{equation}
where $\mathbf{n}$ has i.i.d. entries across the $N$ symbols.
This leads to the following expression for the likelihood of $\mathbf{v}$ as a function of the received signals:
\begin{align}
	\ln
	\mathcal{L} 
	\left( 
	\mathbf{y}
	\mathop{|} 
	\mathbf{v}
	\right)
	=&
	\sum_{n=1}^{N}
	\ln
	\mathcal{L} 
	\left( 
	\mathbf{y}_n
	\mathop{|} 
	\mathbf{v}
	\right) 
	=
	- N_RN \ln \left( \pi \varepsilon^2 \right) \notag  
	\\
	- &
	\sum_{n=1}^{N}
	\frac{1}{\varepsilon^2} \left\| \mathbf{y}_n - \sum_{m=1}^{M}
	\alpha_m \sqrt{p} \mathbf{H}(\theta_m) \mathbf{x}_n \right\|^2,
\end{align}
where $\mathcal{L} \left( \cdot \mathop{|} \cdot \right)$ denotes the likelihood function.
This means that the overall Fisher information from the $N$ observations can be expressed as a sum of contributions from each symbol period.
We use $\mathbf{\Phi}_{\mathsf{o}}$ as an example to illustrate the computation~involved.
The $(m, m^{\prime})$-th entry of $\mathbf{\Phi}_{\mathsf{o}}$ is given by 
\begin{align} \label{BFIM_ind}
	&
	\mathbf{\Phi}_{\mathsf{o}}^{(m,m^{\prime})} 
	=
	- \mathop{\mathbb{E}_{\mathbf{y}, \mathbf{v}}}
	\left[
		\frac{\mathop{\partial}^2}{\partial\theta_m \partial\theta_{m^{\prime}}} \ln \mathcal{L} \left( \mathbf{y}\mathop{|}\mathbf{v} \right)
	\right]  \notag
	\\
	&=
	\frac{2 p N}{\varepsilon^2} \mathop{\cdot}  \notag  
	\\
	& \quad
	\mathfrak{Re}
	\Bigg\lbrace \mathsf{Tr} \Bigg( \mathbb{E}_{\mathbf{v}} \Big[ ( \alpha_{m} \mathbf{\dot{H}}(\theta_m) )^{\mathsf{H}} ( \alpha_{m^{\prime}} \mathbf{\dot{H}}(\theta_{m^{\prime}}) ) \Big]
	\frac{1}{N}
	\sum_{n=1}^{N}
	\mathbf{x}_n \mathbf{x}_n^{\mathsf{H}} \Bigg) \Bigg\rbrace  \notag
	\\
	& =
	\frac{2 p N}{\varepsilon^2}
	\mathfrak{Re}
	\Big\lbrace
	\mathsf{Tr}
	\Big( 
	\mathbb{E}_{\mathbf{v}}
	\Big[  
	( \alpha_{m} \mathbf{\dot{H}}(\theta_m) )^{\mathsf{H}}
	( \alpha_{m^{\prime}} \mathbf{\dot{H}}(\theta_{m^{\prime}}) )
	\Big]
	\mathbf{R}
	\Big) \Big\rbrace,
\end{align}
where $\mathbf{\dot{H}}(\theta_m) \triangleq {\partial \mathbf{H}(\theta_m)} / {\partial \theta_m}$,
and $\mathbf{R}$ is the transmit sample covariance matrix defined as
\begin{equation} 
	\mathbf{R} 
	\triangleq 
	\frac{1}{N} 
	\sum_{n=1}^{N} \mathbf{x}_n \mathbf{x}_n^{\mathsf{H}}.
\end{equation}
It can be observed from \eqref{BFIM_ind} that the BFIMs are the functions of $\mathbf{R}$ only. 
Therefore,
designing the precoding sequence amounts to designing a fixed transmit sample covariance matrix $\mathbf{R}$, 
as done in numerous existing works, e.g., \cite{4276989, 5765721, 4359542, 10138058, 10440056, 9124713, 10147248, 9652071, 8386661, 2509.13661}.

\subsection{With Interfering Scatterers}

\begin{figure*}[!b]
	\hrule
	\begin{align} \label{likelihood}
		\ln \mathcal{L} 
		\left( 
		\mathbf{y}
		\mathop{|} 
		\mathbf{v}
		\right)
		=
		- N_RN \ln \left( \pi \right)
		- \ln \left| \mathbf{\Sigma}(\mathbf{x}) \right|
		-
		\left( \mathbf{y} - \sum_{m=1}^{M}
		\alpha_m \sqrt{p} \mathbf{U}_m \mathbf{x} \right)^{\mathsf{H}}
		\left( \mathbf{\Sigma}(\mathbf{x}) \right)^{-1}
		\left( \mathbf{y} - \sum_{m=1}^{M}
		\alpha_m \sqrt{p} \mathbf{U}_m \mathbf{x} \right)
	\end{align}
\end{figure*}

In the scenario where interfering scatterers are present, the observations are no longer conditionally independent over $N$ given the target parameters, because the interfering scatterers' parameters are the same across the $N$ symbols.
Furthermore, the exact statistics for the computation of BCRLB is complicated.
One way to make the likelihood and the BCRLB tractable is to model 
the aggregate interference-plus-noise as a Gaussian random vector with matching mean and covariance matrix.
Then, the likelihood of $\mathbf{v}$ given the observations $\mathbf{y}$ is expressed in \eqref{likelihood} at the bottom of the page,
where $\mathbf{\Sigma}(\mathbf{x})$ is the covariance matrix of the aggregate interference-plus-noise term,
\begin{equation} \label{Rx}
	\mathbf{\Sigma}(\mathbf{x})
	=
	\sum_{k=1}^{K} 
	p \sigma_k^2
	\mathop{\mathbb{E}_{\eta_k}}
	\left[ 
	\left( \mathbf{\Upsilon}_k \mathbf{x} \right) 
	\left( \mathbf{\Upsilon}_k \mathbf{x} \right)^{\mathsf{H}} 
	\right]
	+
	\varepsilon^2 \mathbf{I}_{N_R N}.
\end{equation}
Then, the BCRLB has the same form as in \eqref{eq:BCRLB}, where
the $(m, m^{\prime})$-th entry of $\mathbf{\Phi}_{\mathsf{o}}$ is given by
\begin{multline} \label{FI}
	\mathbf{\Phi}_{\mathsf{o}}^{(m,m^{\prime})}
	=  \\
	2 p \mathop{\mathfrak{Re}}
	\left\lbrace
	\mathbb{E}_{\mathbf{v}}
	\left[  
	( \alpha_{m} \mathbf{\dot{U}}_m \mathbf{x} )^{\mathsf{H}}
	\left( \mathbf{\Sigma}(\mathbf{x}) \right)^{-1} 
	( \alpha_{m^{\prime}} \mathbf{\dot{U}}_{m^{\prime}} \mathbf{x} )
	\right] \right\rbrace,
\end{multline}
where $\mathbf{\dot{U}}_m \triangleq {\partial \mathbf{U}_m} / {\partial \theta_m}$.
As for the cross-term of BFIM $\mathbf{F}_{\mathsf{o}}$, it consists of two sub-matrices, i.e.,
\begin{equation}
	\mathbf{F}_{\mathsf{o}}
	=
	\left[ 
		\mathfrak{Re} \{ \mathbf{\bar{F}}_{\mathsf{o}} \}, 
	  - \mathfrak{Im} \{ \mathbf{\bar{F}}_{\mathsf{o}} \}
	\right],
\end{equation}
where the $(m,m')$-th entry of $\mathbf{\bar{F}}_{\mathsf{o}}$ is given by
\begin{equation} \label{F_o}
	\mathbf{\bar{F}}_{\mathsf{o}}^{(m,m^{\prime})}
	=
	2 p \mathop{\mathbb{E}_{\mathbf{v}}}
	\left[  
	( \alpha_{m} \mathbf{\dot{U}}_m \mathbf{x} )^{\mathsf{H}}
	\left( \mathbf{\Sigma}(\mathbf{x}) \right)^{-1} 
	( \mathbf{U}_{m^{\prime}} \mathbf{x} )
	\right].
\end{equation}
The expression of $\mathbf{A}_{\mathsf{o}}$ is omitted here, as it does not affect the subsequent derivations.

Here, we remark that the prior information enters the BCRLB in two ways.
First, the observation-induced BFIMs are Bayesian quantities and thus involve expectations over the prior distributions of the unknown parameters, as seen in \eqref{FI} and \eqref{F_o} for the desired signal and in \eqref{Rx} for the interfering scatterers.
Second, the prior information also contributes through the prior BFIMs $\mathbf{\Phi}_{\mathsf{p}}$, $\mathbf{A}_{\mathsf{p}}$, and $\mathbf{F}_{\mathsf{p}}$.

The key observation is that in contrast to the target-only sensing scenario,
the BFIM entries are no longer simple functions of the transmit sample covariance matrix $\mathbf{R}$. 
The BCRLB in fact depends explicitly on the sequence of precoders $\{\mathbf{x}_n\}_{n=1}^{N}$ rather than their second-order sample covariance. 
Thus, unlike the target-only case where it is sufficient to design the transmit sample covariance, the precoder design in the presence of interfering scatterers must be over the sequence of $\{\mathbf{x}_n\}_{n=1}^{N}$.

\section{Precoding for MIMO Sensing with Scatterers}
\label{Section_04}

We now propose a methodology for designing the precoding sequence $\{\mathbf{x}_n\}_{n=1}^{N}$ for MIMO sensing with scatterers.
In this section, we formulate the BCRLB optimization problem for the multiple-target angle estimation. 
The algorithm for solving the problem is presented in the next section.

The first step is to derive the BCRLB for estimating the target azimuth angles $\boldsymbol{\theta}$, which is given by the trace of the top-left principal block of the inverse of full BFIM.
Using the Schur complement, the BCRLB of estimating $\boldsymbol{\theta}$ is given by
\begin{align} \label{schur}
	&
	\mathsf{BCRLB} \left( \boldsymbol{\theta} \right)
	=
	\mathsf{Tr} \bigg\{ \Big( 
	\left( \mathbf{\Phi}_{\mathsf{o}} + \mathbf{\Phi}_{\mathsf{p}} \right)  \notag  \\
	&\hspace{27.5pt}
	-
	\left( \mathbf{F}_{\mathsf{o}} + \mathbf{F}_{\mathsf{p}} \right) 
	\left( \mathbf{A}_{\mathsf{o}} + \mathbf{A}_{\mathsf{p}} \right)^{-1}  
	\left( \mathbf{F}_{\mathsf{o}} + \mathbf{F}_{\mathsf{p}} \right)^{\mathsf{T}} \Big)^{-1} \bigg\}.
\end{align}
The following key assumption, as already mentioned earlier,
allows the above expression to be simplified.
Under the assumption that $\{\alpha_m\}_{m=1}^M$ are independent and distributed as $\mathcal{CN}( 0, \varsigma_m^2 )$,
$\mathbf{F}_{\mathsf{o}}$ is zero due to the expectation over $\boldsymbol{\alpha}$ in \eqref{F_o}.
Similarly, the off-diagonal entries in $\mathbf{\Phi}_{\mathsf{o}}$ are also zero.
Hence, the BCRLB of estimating $\boldsymbol{\theta}$ can be simplified as
\begin{align} \label{BCRLB_eta}
	\mathsf{BCRLB} \left( \boldsymbol{\theta} \right) 
	= 
	\mathsf{Tr}\left\lbrace \left( \mathbf{\Phi}_{\mathsf{o}} + \mathbf{\Phi}_{\mathsf{p}} \right)^{-1} \right\rbrace.
\end{align}
Since both $\mathbf{\Phi}_{\mathsf{o}}$ and $\mathbf{\Phi}_{\mathsf{p}}$ are diagonal, the BCRLB for estimating each single target azimuth angle $\theta_m$ can now be expressed as
\begin{align} \label{BCRLB_m}
	\mathsf{BCRLB} \left( \theta_m \right)
	=
	\frac{1}{\mathbf{\Phi}_{\mathsf{o}}^{(m,m)} + \mathbf{\Phi}_{\mathsf{p}}^{(m,m)}} 
	\triangleq
	\frac{1}{\gamma_m (\mathbf{x})}.
\end{align}

For the multiple-target azimuth angle estimation problem,~we need to determine the relative importance of estimating different targets. 
To ensure fairness, we minimize the maximum BCRLB across all target angles.
Consequently, the problem of designing precoders to optimize the worst-case BCRLB across all angles can be formulated as
\begin{subequations}
	\begin{align}
		\textbf{(P1):} \quad
		\max_{\mathbf{x}} \ \ \min_{m} \;\;
		&\hspace{1pt}
		\gamma_m (\mathbf{x})  \\
		\mathop{\mathrm{subject \ to}} \;\;
		&
		\text{(C1):} \;
		\left| x_i \right| = 1, \ \forall i,  \\
		\mathop{\mathrm{or}} \;\;
		&
		\text{(C2):} \;
		\left\| \mathbf{x}_n \right\|_2 = \sqrt{N_T}, \ \forall n,
	\end{align}
\end{subequations}
which is a max-min fractional programming problem with extra constant-modulus or constant-norm constraints. 

The above precoder design problem depends on the prior distributions of the targets and the scatterers in two ways. 
As mentioned earlier, the priors enter the objective function $\gamma_m(\mathbf{x})$ through both $\mathbf{\Phi}_{\mathsf{p}}^{(m,m)}$ and
$\mathbf{\Phi}_{\mathsf{o}}^{(m,m)}$, where the latter involves expectations
over the priors of the targets and the scatterers, e.g., as in \eqref{Rx} and \eqref{FI}.
In many cases, $\mathbf{\Phi}_{\mathsf{p}}^{(m,m)}$ does not affect the optimization.
For the single-target case, problem (P1) reduces to the maximization of $\gamma_1(\mathbf{x})$. 
Since $\mathbf{\Phi}_{\mathsf{p}}^{(1,1)}$ is independent of $\mathbf{x}$, it can be
omitted when solving for the optimal precoder.
For the multi-target case, if the priors of all the target angles have the same shape and differ only by shifts, then their prior Fisher information is identical, i.e., 
$\mathbf{\Phi}_{\mathsf{p}}^{(m,m)}$ is the same for all $m$. It becomes a
common additive term that can be omitted from the max-min optimization. 

To optimize over $\mathbf{x}$, we need to extract $\mathbf{x}$ from the expectation, which can be done by explicitly computing the expectations via their second-moment matrices. 
The steps below are similar to those of \cite[Theorem~1]{11114787}.  Define 
\begin{align}
	\mathbf{\bar{U}}_m
	\triangleq
	\mathbb{E}_{\theta_m}
	\left[ 
		\mathsf{vec} \left( \mathbf{\dot{U}}_m \right) 
		\mathsf{vec}^{\mathsf{H}} \left( \mathbf{\dot{U}}_m \right) 
	\right],
\end{align}
and let its rank be $R_m$.
Let $\varrho_m^{[r]} $ and $\mathbf{\bar{u}}_m^{[r]}$ be the $r$-th eigenvalue and the corresponding eigenvector of $\mathbf{\bar{U}}_m$.
Then, $\mathbf{\Phi}_{\mathsf{o}}^{(m,m)}$ can be rewritten as
\begin{equation} \label{eFI}
	\mathbf{\Phi}_{\mathsf{o}}^{(m,m)}
	=
	2 p \varsigma_m^2
	\sum_{r=1}^{R_m}
	\varrho_m^{[r]} 
	\left( \mathbf{\dot{U}}_m^{[r]} \mathbf{x} \right)^{\mathsf{H}}
	\left( \mathbf{\Sigma} (\mathbf{x}) \right)^{-1}
	\left( \mathbf{\dot{U}}_m^{[r]} \mathbf{x} \right),
\end{equation}
where with a slight abuse of notation, $\mathbf{\dot{U}}_m^{[r]}$ represents the~matrix reshaped from $\mathbf{\bar{u}}_m^{[r]}$,
i.e., $\mathbf{\dot{U}}_m^{[r]} \triangleq \mathsf{vec}^{-1} ( \mathbf{\bar{u}}_m^{[r]} )$.
Similarly, to rewrite the covariance matrix $\mathbf{\Sigma}(\mathbf{x})$, we define 
\begin{align}
	\mathbf{\bar{\Upsilon}}_k
	\triangleq
	\mathbb{E}_{\eta_k}
	\Big[ 
		\mathsf{vec} \left( \mathbf{\Upsilon}_k \right)
		\mathsf{vec}^{\mathsf{H}} \left( \mathbf{\Upsilon}_k \right)
	\Big],
\end{align}
and let its rank be $T_k$. 
Let $\kappa_k^{[t]}$ and $\boldsymbol{\bar{\upsilon}}_k^{[t]}$ be the $t$-th eigenvalue and the corresponding eigenvector of $\mathbf{\bar{\Upsilon}}_k$.
Then,  $\mathbf{\Sigma} (\mathbf{x})$ can be rewritten as
\begin{equation} \label{decomp}
	\mathbf{\Sigma}(\mathbf{x})
	=
	\sum_{k=1}^{K}
	\sum_{t=1}^{T_k}
	p \sigma_k^2 \kappa_k^{[t]}
	\left( \mathbf{\Upsilon}_k^{[t]} \mathbf{x} \right) 
	\left( \mathbf{\Upsilon}_k^{[t]} \mathbf{x} \right)^{\mathsf{H}} 
	+
	\varepsilon^2 \mathbf{I}_{N_R N},
\end{equation}
where $\mathbf{\Upsilon}_k^{[t]} \triangleq \mathsf{vec}^{-1} (\boldsymbol{\bar{\upsilon}}_k^{[t]})$.
In this way, the precoding sequence design problem (P1) now only involves deterministic quantities.


\section{Algorithm for Optimizing the Precoders}
\label{Section 05}

\subsection{Constant-Norm Linear Transform}

To tackle the max-min fractional programming problem (P1) in an efficient manner,
we adopt a linear transform technique to deal with the multiple fractional sensing metrics $\{\gamma_m (\mathbf{x})\}_{m=1}^{M}$.
This technique is proposed in our prior work \cite{11114787} 
for maximizing a single sum-of-ratios objective in an RIS-assisted ISAC system,
where the optimization variables are RIS reflection coefficients subject to constant-modulus constraints.
In this paper, to design both the phase-only and phase-amplitude precoders,
we show that this linear transform can be generalized to apply to any constraint set that has a constant-norm property. 
We term this technique the \emph{constant-norm linear transform}.

To present the constant-norm linear transform, we begin by defining the following form of fractional functions:
\begin{equation} \label{def}
	f(\mathbf{x})
	\triangleq 
	\left( \mathbf{A} \mathbf{x} \right)^{\mathsf{H}} 
	\left( \mathbf{D}(\mathbf{x})\right)^{-1}
	\left( \mathbf{A} \mathbf{x} \right), 
\end{equation}
where the denominator matrix $\mathbf{D} (\mathbf{x})$ is defined as
\begin{equation}
	\mathbf{D} (\mathbf{x})
	\triangleq  
	\sum_{k} 
		\rho_k 
		(\mathbf{B}_{k}\mathbf{x})
		(\mathbf{B}_{k}\mathbf{x})^{\mathsf{H}} 
	+ 
	\mathbf{C},
\end{equation}
the matrix $\mathbf{C}$ is positive definite, $\rho_k >0$ for all $k$, 
and $\mathbf{x}$ is subject to either (C1) or (C2) as in \eqref{C1} and \eqref{C2}.

The key observation is under either the constraint (C1) or (C2), we have the following \emph{constant-norm} property:
\begin{equation}
	\|\mathbf{x}\|_2 = \sqrt{N_T N}.
	\label{eq:constant-norm-constraint}
\end{equation}

The following lemma extends \cite[Lemma~2]{11114787} to fractional functions defined over such constant-norm constraints. 
\begin{lemma}[Constant-Norm Linear Transform]
	\label{Lemma_1}
	Consider a fractional function $f(\mathbf{x})$ over $\mathbf{x} \in \mathcal{X}$, where the constraint set $\mathcal{X}$ satisfies the constant-norm property (\ref{eq:constant-norm-constraint}). Then, 
	\begin{align} \label{LT_F}
		f(\mathbf{x}) 
		\geq 
		\bar{f}(\mathbf{x}, \mathbf{z}, \mathbf{g})
		\triangleq
		2 \mathop{\mathfrak{Re}}\left\lbrace \mathbf{x}^{\mathsf{H}} \mathbf{u} (\mathbf{z}, \mathbf{g}) \right\rbrace
		+ 
		c (\mathbf{z}, \mathbf{g}),
	\end{align}
	for all $\mathbf{g}$ and for all $\mathbf{x}, \mathbf{z} \in \mathcal{X}$,
	where the coefficient of the linear term in $\mathbf{x}$ is
	\begin{equation}
		\mathbf{u} (\mathbf{z}, \mathbf{g})
		=
		\left( \delta \mathbf{I} - \mathbf{M} \right) \mathbf{z} 
		+ 
		\mathbf{A}^{\mathsf{H}} \mathbf{g},
	\end{equation}
	the matrix $\mathbf{M}$ is given by
	\begin{equation}
		\mathbf{M}
		= 
		\sum_{k}
			\rho_k 
			\big(\mathbf{B}_{k}^{\mathsf{H}} \mathbf{g}\big)
			\big(\mathbf{B}_{k}^{\mathsf{H}} \mathbf{g}\big)^{\mathsf{H}},
	\end{equation}
	the parameter $\delta$ is the trace of $\mathbf{M}$, and $c (\mathbf{z}, \mathbf{g})$ is given by
	\begin{equation}
		c (\mathbf{z}, \mathbf{g})
		=
		\mathbf{z}^{\mathsf{H}} \mathbf{M} \mathbf{z} 
		- 
		2 \delta N_T N
		- 
		\mathbf{g}^{\mathsf{H}} \mathbf{C} \mathbf{g} .
	\end{equation}
	The equality in \eqref{LT_F} is achieved at
	\begin{align} \label{opt_auxi}
		\mathbf{z}^{\star} 
		= 
		\mathbf{x},
		\ \text{and} \
		\mathbf{g}^{\star}  
		= 
		(\mathbf{D} (\mathbf{x}))^{-1} \mathbf{A} \mathbf{x}.
	\end{align}
\end{lemma}

\begin{IEEEproof}
The proof follows the same arguments as in \cite[Lemma~2]{11114787} and is presented in Appendix~\ref{app: Lemma_01}.
\end{IEEEproof}

Lemma~\ref{Lemma_1} yields the following theorem, which provides 
a reformulation of a class of optimization problems involving ratios.

\begin{theorem} \label{theorem_1}
	Suppose that $h_{\mathsf{o},m}(\cdot)$ and $h_{\mathsf{c},w} (\cdot)$ are componentwise nondecreasing, and $\mathbf{f}_{\mathsf{o},m}(\mathbf{x})$ and $\mathbf{f}_{\mathsf{c},w}(\mathbf{x})$ are vector-valued functions defined as
	\begin{align}
		\mathbf{f}_{\mathsf{o},m}(\mathbf{x})
		& \triangleq	
		\left[ f_{\mathsf{o},m}^{[1]}(\mathbf{x}), 
   		\ldots, 
   		f_{\mathsf{o},m}^{[R_m]}(\mathbf{x}) \right],  
		\\
		\mathbf{f}_{\mathsf{c},w}(\mathbf{x})
		& \triangleq	
		\left[ f_{\mathsf{c},w}^{[1]}(\mathbf{x}), 
      	\ldots, 
     	f_{\mathsf{c},w}^{[T_w]}(\mathbf{x}) \right],
	\end{align}
	where $f_{\mathsf{o},m}^{[r]}(\mathbf{x})$ and $f_{\mathsf{c},w}^{[t]}(\mathbf{x})$ are ratios of the form \eqref{def}.
	Then, the following fractional program
	\begin{subequations} \label{max-min-ratios}
		\begin{align}
			\max_{\mathbf{x}} \hspace{7pt} \min_{m} \;\;
			&\;
			h_{\mathsf{o},m}\big( \mathbf{f}_{\mathsf{o},m}(\mathbf{x}) \big)  \\
			\mathop{\mathrm{subject \,\, to}} \;\;
			&\;
			h_{\mathsf{c},w}\big( \mathbf{f}_{\mathsf{c},w}(\mathbf{x}) \big) \geq \Gamma_w, \ \forall w,  \\
			&\;
			\mathbf{x} \in \mathcal{X}, 
		\end{align}
	\end{subequations}
	where the constraint set $\mathcal{X}$ satisfies the constant-norm property (\ref{eq:constant-norm-constraint}), is equivalent to
	\begin{subequations} \label{newmax-minLP}
		\begin{align}
			\max_{
				\substack{\mathbf{x}, \mathbf{z} \\ \{\mathbf{G}_{\mathsf{o},m}\}, \{\mathbf{G}_{\mathsf{c},w}\}}} 
			\min_{m} \;\;
			&\;
			h_{\mathsf{o},m} \big(\bar{\mathbf{f}}_{\mathsf{o},m}(\mathbf{x}, \mathbf{z}, \mathbf{G}_{\mathsf{o},m})\big)
			\\
			\mathop{\mathrm{subject\ to}} \quad \ \
			&\;
			h_{\mathsf{c},w} \big(\bar{\mathbf{f}}_{\mathsf{c},w} (\mathbf{x}, \mathbf{z}, \mathbf{G}_{\mathsf{c},w})\big) \geq \Gamma_w, \ \forall w,  \\
			&\;
			\mathbf{x}, \mathbf{z} \in \mathcal{X},
		\end{align}
	\end{subequations}
	where $\mathbf{G}_{\mathsf{o},m} = [\mathbf{g}_m^{[1]}, \ldots, \mathbf{g}_m^{[{R_m}]}]$ and $\mathbf{G}_{\mathsf{c},w} = [\mathbf{g}_m^{[1]}, \ldots, \mathbf{g}_m^{[{R_m}]}]$.
	The vectors $\bar{\mathbf{f}}_{\mathsf{o},m}(\mathbf{x}, \mathbf{z}, \mathbf{G}_{\mathsf{o},m})$ and $\bar{\mathbf{f}}_{\mathsf{c},w}(\mathbf{x}, \mathbf{z}, \mathbf{G}_{\mathsf{c},w})$ are given~by
	\begin{align}
		& \bar{\mathbf{f}}_{\mathsf{o},m}(\mathbf{x}, \mathbf{z}, \mathbf{G}_{\mathsf{o},m})
		\triangleq  \notag  \\
		& \qquad
		\left[ \bar{f}_{\mathsf{o},m}^{[1]}(\mathbf{x}, \mathbf{z}, \mathbf{g}_{\mathsf{o},m}^{[1]}), \ldots, \bar{f}_{\mathsf{o},m}^{[R_m]}(\mathbf{x}, \mathbf{z}, \mathbf{g}_{\mathsf{o},m}^{[R_m]}) \right],  \\
		& \bar{\mathbf{f}}_{\mathsf{c},w}(\mathbf{x}, \mathbf{z}, \mathbf{G}_{\mathsf{c},w})
		\triangleq  \notag  \\
		& \qquad
		\left[ \bar{f}_{\mathsf{c},w}^{[1]}(\mathbf{x}, \mathbf{z}, \mathbf{g}_{\mathsf{c},w}^{[1]}), \ldots, \bar{f}_{\mathsf{c},w}^{[T_w]}(\mathbf{x}, \mathbf{z}, \mathbf{g}_{\mathsf{c},w}^{[T_w]}) \right],
	\end{align}
	where $\bar{f}_{\mathsf{o},m}^{[r]}(\mathbf{x}, \mathbf{z}, \mathbf{g}_{\mathsf{o},m}^{[r]})$ and $\bar{f}_{\mathsf{c},w}^{[t]}(\mathbf{x}, \mathbf{z}, \mathbf{g}_{\mathsf{c},w}^{[t]})$
	are the transformed functions of
	$f_{\mathsf{o},m}^{[r]}(\mathbf{x})$ and $f_{\mathsf{c},w}^{[t]}(\mathbf{x})$, respectively,
	as defined in the linear transform \eqref{LT_F} in Lemma~\ref{Lemma_1}. 
\end{theorem}

\begin{IEEEproof}
	The proof is given in Appendix~\ref{app: Theorem_01}.
\end{IEEEproof}

Theorem~\ref{theorem_1} enables the original fractional program \eqref{max-min-ratios} to be solved through its linear transformed problem \eqref{newmax-minLP}.
A natural approach for solving \eqref{newmax-minLP} is to alternately optimize the auxiliary variables $\{\mathbf{G}_{\mathsf{o},m}\}$, $\{\mathbf{G}_{\mathsf{c},w}\}$, and $\mathbf{z}$, and the original variable $\mathbf{x}$.
For fixed $\mathbf{x}$, the optimal auxiliary variables can be obtained in closed form as in \eqref{opt_auxi}.
For fixed auxiliary variables, the resulting subproblem over $\mathbf{x}$ in \eqref{newmax-minLP} can often be handled via its Lagrangian dual problem, as shown in the next section.

\subsection{Solving Problem (P1)}

We now solve problem (P1) using the constant-norm linear transform.
Applying Theorem~\ref{theorem_1}, (P1) can be equivalently transformed into the following problem:
\begin{subequations}
	\begin{align}
		\textbf{(P2):} \quad
		\max_{\mathbf{x}, \mathbf{z}, \{\mathbf{G}_m\}} \min_{m} \;\;
		&\hspace{1pt}
		\bar{\gamma}_m (\mathbf{x}, \mathbf{z}, \mathbf{G}_m)  \\
		\mathop{\mathrm{subject\;\;to}} \;\;
		&
		\text{(C1):} \;
		\left| x_i \right| = \left| z_i \right| = 1, \ \forall i, \\
		\mathop{\mathrm{or}} \;\;
		&
		\text{(C2):} \;
		\left\| \mathbf{x}_n \right\|_2 = \left\| \mathbf{z}_n \right\|_2 = \sqrt{N_T}, \ \forall n.
	\end{align}
\end{subequations}
The transformed objective function $\bar{\gamma}_m (\mathbf{x}, \mathbf{z}, \mathbf{G}_m) $ is given by
\begin{multline} \label{linear_m}
	\bar{\gamma}_m (\mathbf{x}, \mathbf{z}, \mathbf{G}_m) 
	=
	4 p \varsigma_m^2
	\mathop{\mathfrak{Re}}
	\Big\lbrace
	\mathbf{x}^{\mathsf{H}} \mathbf{u}_m (\mathbf{z}, \mathbf{G}_m)
	\Big\rbrace
	\\
	+
	2 p \varsigma_m^2
	c_m(\mathbf{z}, \mathbf{G}_m) 
	+
	\mathbf{\Phi}_{\mathsf{p}}^{(m,m)},
\end{multline}
where $\mathbf{G}_m  \triangleq [\mathbf{g}_m^{[1]}, \ldots, \mathbf{g}_m^{[{R_m}]}]$,
and the linear coefficient of $\mathbf{x}$ in \eqref{linear_m} is given by
\begin{equation}
	\mathbf{u}_m (\mathbf{z}, \mathbf{G}_m)
	=
	\sum_{r=1}^{R_m}
	\varrho_m^{[r]}
	\left(
	\big( \delta_m^{[r]} \mathbf{I} - \mathbf{M}_m^{[r]} \big) \mathbf{z} 
	+ 
	\big( \mathbf{\dot{U}}_m^{[r]} \big)^{\mathsf{H}} \mathbf{g}_m^{[r]}
	\right),
\end{equation}
the matrix $\mathbf{M}_m^{[r]}$ is given by
\begin{equation}
	\mathbf{M}_m^{[r]}
	=
	\sum_{k=1}^{K}
	\sum_{t=1}^{T_k}
	p \sigma_k^2 \kappa_k^{[t]}
	\left( \big( \mathbf{g}_m^{[r]} \big)^{\mathsf{H}} \mathbf{\Upsilon}_k^{[t]} \right)^{\mathsf{H}}
	\left( \big( \mathbf{g}_m^{[r]} \big)^{\mathsf{H}} \mathbf{\Upsilon}_k^{[t]} \right),
\end{equation}
the parameter $\delta_m^{[r]}$ denotes the trace of the positive semi-definite matrix $\mathbf{M}_m^{[r]}$,
and $c_m (\mathbf{z}, \mathbf{G}_m)$ is given by
\begin{equation}
	c_m (\mathbf{z}, \mathbf{G}_m)
	=
	\sum_{r=1}^{R_m}
	\varrho_m^{[r]}
	\left( 
	\mathbf{z}^{\mathsf{H}} \mathbf{M}_m^{[r]} \mathbf{z} 
	- 
	2 \delta_m^{[r]} N_T N
	- 
	\varepsilon^2 \big\| \mathbf{g}_m^{[r]} \big\|_2^2
	\right).
\end{equation}

Problem (P2) can now be solved in an iterative manner. When $\mathbf{x}$ is held fixed, the optimal $\mathbf{z}^{\star}$ and $(\mathbf{g}_m^{[r]})^{\star}$ are given by
\begin{align}
	\mathbf{z}^{\star} 
	= 
	\mathbf{x},  
	\ \text{and} \
	\big( \mathbf{g}_m^{[r]} \big)^{\star}
	= 
	(\mathbf{\Sigma} (\mathbf{x}))^{-1} \mathbf{\dot{U}}_m^{[r]} \mathbf{x}. \label{54}
\end{align}
When $\mathbf{z}$ and $\mathbf{g}_m^{[r]}$ are held fixed, updating $\mathbf{x}$ requires us to solve
the following problem:
\begin{subequations}
	\begin{align}
		\textbf{(P2-1):} \quad
		\max_{\mathbf{x}} \ \ \min_{m} \;\;
		&\hspace{1pt}
		\bar{\gamma}_m (\mathbf{x}, \mathbf{z}, \mathbf{G}_m)  \\
		\mathop{\mathrm{subject \ to}} \;\;
		&
		\text{(C1):} \;
		\left| x_i \right| = 1, \ \forall i, \label{cmconst} \\
		\mathop{\mathrm{or}} \;\;
		&
		\text{(C2):} \;
		\left\| \mathbf{x}_n \right\|_2 = \sqrt{N_T}, \ \forall n.  \label{cnconst}
	\end{align}
\end{subequations}
Solving the above problem over the nonconvex precoding constraints may still be challenging.
Next, we show that (P2-1) can be solved to global optimality with low complexity,
provided that a specific condition is satisfied.

\subsection{Global Optimality of Problem (P2-1)}

To solve (P2-1) efficiently, we relax the phase-only precoding constraint \eqref{cmconst} into $\left| x_i \right| \leq 1$, 
and relax the phase-amplitude precoding constraint \eqref{cnconst} into $\left\| \mathbf{x}_n \right\|_2 \leq \sqrt{N_T}$.
Consequently, the relaxed version of (P2-1) is given by
{\mathtoolsset{showonlyrefs=false}
\begin{subequations} \label{relaxedPro}
	\begin{align}
		\max_{\mathbf{x}} \ \ \min_{m} \;\;
		&\hspace{1pt}
		\bar{\gamma}_m (\mathbf{x}, \mathbf{z}, \mathbf{G}_m)   \label{38a} \\
		\mathop{\mathrm{subject \ to}} \;\;
		&
		\text{(C1'):} \;
		\left| x_i \right| \leq 1, \ \forall i,  \\
		\mathrm{or} \;\;
		&
		\text{(C2'):} \;
		\left\| \mathbf{x}_n \right\|_2 \leq \sqrt{N_T}, \ \forall n.
	\end{align}
\end{subequations}}Note that the optimal objective value of the relaxed problem (\ref{relaxedPro}) is an \emph{upper bound} to that of problem (P2-1) as the relaxed problem (\ref{relaxedPro}) has a larger feasible set than problem (P2-1).

Next, observe that the problem \eqref{relaxedPro} is convex and the strong duality holds.
Consider its dual problem, which is given by
\begin{subequations} \label{min-max}
	\begin{align}
		\min_{\boldsymbol{\nu}} \hspace{6.5pt} \max_{\mathbf{x}}  \;\;
		&\hspace{1pt}
		L ( \mathbf{x}, \boldsymbol{\nu} )
				\triangleq
				\sum_{m=1}^{M} \nu_m
				\bar{\gamma}_m (\mathbf{x}, \mathbf{z}, \mathbf{G}_m)  \label{36a} \\
		\mathrm{subject\;to} \;\;
		&\hspace{1pt}
		\boldsymbol{\nu} \in \Delta^{M-1},  \\
		&
		\text{(C1'):} \;
		\left| x_i \right| \leq 1, \ \forall i, \\
		\mathrm{or} \;\;
		&
		\text{(C2'):} \;
		\left\| \mathbf{x}_n \right\|_2 \leq \sqrt{N_T}, \ \forall n,
	\end{align}
\end{subequations}
where $\Delta^{M-1}$ denotes the standard simplex in $\mathbb{R}^M$,
and $\boldsymbol{\nu}$ is the dual variable.

For fixed $\boldsymbol{\nu}$, the inner maximization in \eqref{36a} can be solved in closed form.
Note that the objective function is linear in $\mathbf{x}$ with coefficient
\begin{equation} 
	\boldsymbol{\pi} (\boldsymbol{\nu})
	\triangleq
	\sum_{m=1}^{M} 
	\nu_m  \varsigma_m^2 \mathbf{u}_m (\mathbf{z}, \mathbf{G}_m).
\end{equation}
For the phase-only case, the optimal $\mathbf{x}$ for the inner maximization problem is obtained by setting its entries to match the phases of $\boldsymbol{\pi} (\boldsymbol{\nu})$, i.e., 
\begin{equation} \label{37}
	\mathbf{x}^{\star} (\boldsymbol{\nu})
	=
	\left[
		e^{j \angle{[\boldsymbol{\pi}(\boldsymbol{\nu})]_1}}, 
		\ldots,
		e^{j \angle{[\boldsymbol{\pi}(\boldsymbol{\nu})]_{N_TN}}} 
	\right]^\mathsf{T},
\end{equation}
For the phase-amplitude case, the inner maximizer is obtained by normalizing each length-$N_T$ block of $\boldsymbol{\pi} (\boldsymbol{\nu})$ as follows:
\begin{equation} \label{dig_sol}
	\mathbf{x}^{\star} (\boldsymbol{\nu})
	=
	\left[ 
	\sqrt{N_T}
	\frac{\boldsymbol{\pi}_1^{\mathsf{T}} (\boldsymbol{\nu})}{\left\| \boldsymbol{\pi}_1 (\boldsymbol{\nu}) \right\|_2},
	\ldots,
	\sqrt{N_T}
	\frac{\boldsymbol{\pi}_N^{\mathsf{T}} (\boldsymbol{\nu})}{\left\| \boldsymbol{\pi}_N (\boldsymbol{\nu}) \right\|_2}
	\right]^{\mathsf{T}},
\end{equation}
where $\boldsymbol{\pi}_n (\boldsymbol{\nu})$ denotes the $n$-th block vector of $\boldsymbol{\pi} (\boldsymbol{\nu})$.

Substituting $\mathbf{x}^{\star} (\boldsymbol{\nu})$ into (\ref{36a}), the problem \eqref{min-max} can be rewritten as follows:
\begin{subequations}
	\begin{align}
		\textbf{(P2-2):} \quad
		\mathop{\mathrm{minimize}}_{\boldsymbol{\nu}}  \;\;
		& \;
		L \left( \mathbf{x}^{\star} (\boldsymbol{\nu}), \boldsymbol{\nu} \right)  \\
		\mathop{\mathrm{subject \, to}} \;\;
		&\;
		\boldsymbol{\nu} \in \Delta^{M-1},	
	\end{align}
\end{subequations}
where for the phase-only precoding case, the objective is
\begin{align} \label{obj_36}
	L ( \mathbf{x}^{\star} (\boldsymbol{\nu}), \boldsymbol{\nu} )
	& = 
	4 p
	\left\| \boldsymbol{\pi} (\boldsymbol{\nu}) \right\|_{1}  \notag  \\
	& + 
	\sum_{m=1}^{M} \nu_m \Big( 2 p \varsigma_m^2
	c_m (\mathbf{z}, \mathbf{G}_m) 
	+
	\mathbf{\Phi}_{\mathsf{p}}^{(m,m)} \Big),
\end{align}
and for the phase-amplitude precoding case, the objective is
\begin{align} \label{obj_2}
	L ( \mathbf{x}^{\star} (\boldsymbol{\nu}), \boldsymbol{\nu} )
	& = 
	4 p \sqrt{N_T}
	\sum_{n=1}^{N}
	\left\| \boldsymbol{\pi}_n (\boldsymbol{\nu}) \right\|_{2}  \notag  \\
	& + 
	\sum_{m=1}^{M} \nu_m \left( 2 p \varsigma_m^2
	c_m (\mathbf{z}, \mathbf{G}_m) 
	+
	\mathbf{\Phi}_{\mathsf{p}}^{(m,m)} \right).
\end{align}
Problem (P2-2) is convex in $\boldsymbol{\nu}$ with dimension $M \ll N_T N$.
Hence, $\boldsymbol{\nu}^{\star}$ can be obtained very efficiently, and $\mathbf{x}^{\star} ( \boldsymbol{\nu}^{\star} )$ can be obtained from $\boldsymbol{\nu}^\star$ via \eqref{37} and \eqref{dig_sol}.

\begin{algorithm}[t]
	\caption{\strut Precoder Design for Multi-Target MIMO Sensing with Signal-Dependent Interfering Scatterers}
	\begin{algorithmic}[1] \label{Algorithm_01}
		
		\vspace{3pt}
		\STATE Initialize $\mathbf{x}$.
		
		\REPEAT
		
		\STATE Update $\mathbf{z} = \mathbf{x}$.
		
		\STATE Update $\mathbf{g}_m^{[r]}$ by \eqref{54} for all $r$ and $m$.
		
		\STATE Update $\boldsymbol{\nu}$ by solving the convex problem (P2-2).
		
		\STATE Update $\mathbf{x}$ by \eqref{37} or \eqref{dig_sol}.
		
		\STATE \hspace{-8.4pt}* Update $\mathbf{x}$ with trial extrapolation $\mathbf{x}_{\mathsf{try}}$ if better.
		
		\UNTIL{\textbf{convergence}}
		
		\STATE \textbf{output} Optimized $\mathbf{x}$. 
	\end{algorithmic}
\end{algorithm}

However, there is a crucial caveat in this proposed approach. Observe that
obtaining $\mathbf{x}^{\star} ( \boldsymbol{\nu}^{\star} )$ based on \eqref{37} and \eqref{dig_sol} works only if the following conditions hold 
\begin{align}
	\text{phase-only:} \;
	&
	\left[ \boldsymbol{\pi} (\boldsymbol{\nu}^{\star}) \right]_i \neq 0, \hspace{6.5pt} i = 1, 2, \ldots, N_TN, \label{condition} \\ 
	\text{phase-amplitude:} \;
	& \hspace{5pt}
	\boldsymbol{\pi}_n (\boldsymbol{\nu}^{\star}) \neq \mathbf{0}, \ n = 1,2, \ldots, N. \label{condition2}
\end{align}
In this case, 
the solution $\mathbf{x}^{\star} (\boldsymbol{\nu}^{\star})$, given by
\eqref{37} for the phase-only precoding case and by \eqref{dig_sol} for the phase-amplitude precoding case,
is the \emph{unique} maximizer of $L ( \mathbf{x}, \boldsymbol{\nu}^{\star} )$.
Otherwise, the optimal $\mathbf{x}^{\star} ( \boldsymbol{\nu}^{\star} )$ for fixed $\boldsymbol{\nu}^{\star}$ is not unique.

For the relaxed problem \eqref{relaxedPro}, which is convex,
whenever $\mathbf{x}^{\star} (\boldsymbol{\nu}^{\star})$ corresponding to the optimal dual variable $\boldsymbol{\nu}^{\star}$ is unique,
it must be the optimal solution of the primal optimization problem \eqref{relaxedPro}. 
This can be justified by \cite[Corollary 28.1.1]{rockafellar1970convex}.

Furthermore, in this case, since \eqref{relaxedPro} is a relaxed version of (P2-1) and the optimal $\mathbf{x}^{\star}(\boldsymbol{\nu}^{\star})$ already satisfies
the constant-modulus and constant-norm constraints for the phase-only
and phase-amplitude cases, respectively, it must also be 
an optimal solution of the original nonconvex problem (P2-1).

We summarize the above result in the following theorem.
\begin{theorem}
	An optimal solution to problem (P2-1) with phase-only constraint (C1) can be obtained by solving problem (P2-2) with objective function \eqref{obj_36}, if its optimal solution of $\boldsymbol{\nu}^{\star}$ as in \eqref{37} satisfies the condition \eqref{condition}. 
	The same conclusion holds for the phase-amplitude case with the objective function of problem (P2-2) replaced by \eqref{obj_2}, the condition replaced by \eqref{condition2}, in which case the solution $\mathbf{x}^{\star}(\boldsymbol{\nu}^{\star})$ as in \eqref{dig_sol} is an optimal solution to problem (P2-1) with constraint (C2).
\end{theorem}

Empirically, for the scenarios considered in this paper, where the dimension of $\mathbf{x}$ is much greater than the number of targets, the condition \eqref{condition} or \eqref{condition2} is almost always satisfied.

However, if the condition \eqref{condition} or \eqref{condition2} is not satisfied,
the recovery of the primal optimal variables from the optimal dual would be challenging.

Finally, if the condition \eqref{condition} or \eqref{condition2} is satisfied at every iteration, the proposed algorithm is guaranteed to converge,
and the converged solution is a stationary point of the original problem (P1).
The details are shown in Algorithm~\ref{Algorithm_01}.

We remark that the convergence rate of the proposed algo\-rithm can be accelerated using momentum-based methods \cite{10605808}.
Specifically, at each iteration, 
we try a momentum-based extrapolated point $\mathbf{x}_{\mathsf{try}}$. It is accepted if it yields a better objective value; otherwise, the update of $\mathbf{x}$ by \eqref{37} or \eqref{dig_sol} is retained.

\section{Joint Precoding and Combining Design}
\label{Section_06}

In the previous sections, we consider a MIMO sensing system equipped with a fully digital receiver, where only the precoder needs to be optimized.
In this section, we further consider the case where the receiver is implemented as an analog combiner followed by single-RF chain processing, as illustrated in Fig.~\ref{Figure_02}, and the precoder and the combiner need to be jointly optimized.

\begin{figure}[!t]
	\centering
	\includegraphics[width = 1\linewidth]{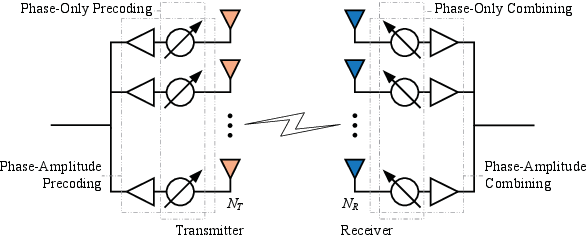}
	\caption{MIMO sensing with joint precoding and combining.}
	\label{Figure_02}
\end{figure}

Similar to the precoder $\mathbf{x}$,
the sequence of receive combiners across $N_R$ antennas and $N$ symbol periods can be expressed~as
\begin{align}
	\mathbf{w}
	\triangleq
	\mathsf{vec}
	\left(
	\left[ 
	\mathbf{w}_1, 
	\mathbf{w}_2, 
	\ldots, 
	\mathbf{w}_N
	\right]
	\right),
\end{align}
where $\mathbf{w}_n \in \mathbb{C}^{N_R \times 1}$ denotes the receive combiner designed~for the $n$-th symbol period.
We again consider two cases, the phase-only combiner and phase-amplitude combiner, for which $\mathbf{w}$ satisfies their respective constraints, as below:
\begin{align}
	& \text{(C3):}
	\quad
	\left| w_{i} \right| 
	= 1, 
	\quad
	i = 1, 2, \ldots, N_RN,  \\
	& \text{(C4):}
	\quad
	\left\| \mathbf{w}_n \right\|_2 
	= 
	\sqrt{N_R}, 
	\quad
	n = 1, 2, \ldots, N.
\end{align}

At each symbol period, the echoes received across the $N_R$ antennas are combined into a scalar.
The resulting $N$ scalars are stacked into a vector $\mathbf{r} \in \mathbb{C}^{N \times 1}$ as
\begin{align} \label{56}
	\mathbf{r}
	& =
	\mathbf{\bar{W}} \mathbf{y}  \notag  \\
	& =
	\sum_{m=1}^{M}
	\alpha_m \sqrt{p}
	\mathbf{\bar{W}} \mathbf{U}_m \mathbf{x}
	+
	\sum_{k=1}^{K}
	\beta_k \sqrt{p}
	\mathbf{\bar{W}} \mathbf{\Upsilon}_k \mathbf{x}
	+
	\mathbf{\bar{n}},
\end{align}
and can be equivalently rewritten in terms of the vectorized combiner $\mathbf{w}$ as
\begin{align} \label{57}
	\mathbf{r}
	=
	\sum_{m=1}^{M}
	\alpha_m \sqrt{p}
	\mathbf{\bar{X}} \mathbf{U}_m^{\mathsf{T}} \mathbf{w}
	+
	\sum_{k=1}^{K}
	\beta_k \sqrt{p} \mathbf{\bar{X}} \mathbf{\Upsilon}_k^{\mathsf{T}} \mathbf{w}
	+
	\mathbf{\bar{n}},
\end{align}
where the block-diagonal matrices $\mathbf{\bar{W}}$ and $\mathbf{\bar{X}}$ are defined as
\begin{align}
	\mathbf{\bar{W}}
	\triangleq
	\begin{bmatrix}
		\mathbf{w}_1^{\mathsf{T}} & \cdots & \mathbf{0} \\
		\vdots & \ddots & \vdots \\
		\mathbf{0} & \cdots & \mathbf{w}_N^{\mathsf{T}}
	\end{bmatrix},
	\ \text{and} \
	\mathbf{\bar{X}}
	\triangleq
	\begin{bmatrix}
		\mathbf{x}_1^{\mathsf{T}} & \cdots & \mathbf{0} \\
		\vdots & \ddots & \vdots \\
		\mathbf{0} & \cdots & \mathbf{x}_N^{\mathsf{T}}
	\end{bmatrix},
\end{align}
and the additive noise $\mathbf{\bar{n}} \sim \mathcal{CN} (0, \varepsilon^2 N_{R}\mathbf{I}_{N})$.

Following the same derivations as in Sections~\ref{Section_03} and \ref{Section_04}, the sensing performance metric for estimating $\theta_m$ can be derived in two different forms, either from \eqref{56} as a function of the vectorized precoder $\mathbf{x}$:
\begin{align} \label{metric_analog}
	& \gamma_m (\mathbf{x}, \mathbf{w})
	\triangleq  \mathbf{\Phi}_{\mathsf{p}}^{(m,m)} +  \notag  \\
	& 2 p \varsigma_m^2
	\sum_{r=1}^{R_m}
	\varrho_m^{[r]}
	\left( \mathbf{\bar{W}} \mathbf{\dot{U}}_m^{[r]} \mathbf{x} \right)^{\mathsf{H}}
	\left( \mathbf{\Sigma} (\mathbf{x}, \mathbf{w}) \right)^{-1}
	\left( \mathbf{\bar{W}} \mathbf{\dot{U}}_m^{[r]} \mathbf{x} \right),
\end{align}
where $\mathbf{\Sigma} (\mathbf{x}, \mathbf{w})$ is the covariance of the interference-plus-noise term, expressed as
\begin{multline}
	\mathbf{\Sigma}(\mathbf{x}, \mathbf{w})
	=
	\sum_{k=1}^{K}
	\sum_{t=1}^{T_k}
	p \sigma_k^2 \kappa_k^{[t]}
	\left( \mathbf{\bar{W}} \mathbf{\Upsilon}_k^{[t]} \mathbf{x} \right) \left( \mathbf{\bar{W}} \mathbf{\Upsilon}_k^{[t]} \mathbf{x} \right)^{\mathsf{H}} \\
	+
	\varepsilon^2 N_{R}\mathbf{I}_{N},
\end{multline}
or based on \eqref{57}, the metric $\gamma_m (\mathbf{x}, \mathbf{w})$ can be equivalently rewritten as a function of the vectorized combiner $\mathbf{w}$:
\begin{equation}
	\begin{aligned} \label{61}
		&
		\gamma_m (\mathbf{x}, \mathbf{w})
		\triangleq  \mathbf{\Phi}_{\mathsf{p}}^{(m,m)} +  \\
		&
		2 p \varsigma_m^2
		\sum_{r=1}^{R_m}
		\varrho_m^{[r]} 
		\left( \mathbf{\bar{X}} \big( \mathbf{\dot{U}}_m^{[r]} \big)^{\mathsf{T}} \mathbf{w} \right)^{\mathsf{H}}
		\left( \mathbf{\Sigma} (\mathbf{x}, \mathbf{w}) \right)^{-1}
		\left( \mathbf{\bar{X}} \big( \mathbf{\dot{U}}_m^{[r]} \big)^{\mathsf{T}} \mathbf{w} \right), 
	\end{aligned}
\end{equation}
where $\mathbf{\Sigma} (\mathbf{x}, \mathbf{w})$ can also be equivalently rewritten as
\begin{multline}
	\mathbf{\Sigma}(\mathbf{x}, \mathbf{w})
	=
	\sum_{k=1}^{K}
	\sum_{t=1}^{T_k}
	p \sigma_k^2 \kappa_k^{[t]}
	\left( \mathbf{\bar{X}} \big( \mathbf{\Upsilon}_k^{[t]} \big)^{\mathsf{T}} \mathbf{w} \right) 
	\left( \mathbf{\bar{X}} \big( \mathbf{\Upsilon}_k^{[t]} \big)^{\mathsf{T}} \mathbf{w} \right)^{\mathsf{H}}  \\
	+
	\varepsilon^2 N_{R}\mathbf{I}_{N}.
\end{multline}

The joint precoding and combining design problem can then be formulated as follows using either of the two forms above:
\begin{subequations}
	\begin{align}
		\textbf{(P3):} \quad
		\max_{\mathbf{x}, \mathbf{w}} \ \ \min_{m} \;\;
		&\hspace{1pt}
		\gamma_m (\mathbf{x}, \mathbf{w})  \\
		\mathop{\mathrm{subject \ to}} \;\;
		&
		\text{(C5):} \;\;
		\left| x_i \right| = 1, \ \left| w_{j} \right| = 1, \ \forall i, j,  \\
		\mathop{\mathrm{or}} \;\;
		&
		\text{(C6):} \;\;
		\left\| \mathbf{x}_n \right\|_2 = \sqrt{N_T}, \ \forall n,  \\
		&
		\hspace{24.5pt}
		\left\| \mathbf{w}_{n} \right\|_2 = \sqrt{N_R}, \ \forall n.
	\end{align}
\end{subequations}
Problem (P3) is a multi-variable max-min fractional programming problem with constant-modulus/norm constraints.
This problem can be tackled using alternating optimization combined with the linear transform.
Specifically, for fixed $\mathbf{w}$, the linear-transform-based algorithm developed in Section~\ref{Section 05} is applied to the objective \eqref{metric_analog} to update $\mathbf{x}$; for fixed $\mathbf{x}$, the same algorithm is applied to the objective \eqref{61} to update $\mathbf{w}$.
The detailed update procedures are similar to Algorithm~\ref{Algorithm_01}.

\section{Multi-Stage Sensing with Scatterers}
\label{Section IV}

The previous sections focus on sensing over a single stage, where $N$ observations are collected to carry out one round of estimation.
However, in practical scenarios, sensing is often performed sequentially across multiple stages, with each stage consisting of multiple symbol periods.
In this section, we extend the proposed one-stage framework to a multi-stage setting.

\begin{figure}[!t]
	\centering
	\includegraphics[width = 1\linewidth]{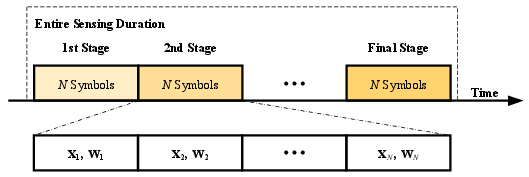}
	\caption{Multi-stage sensing with adaptive precoding and combining.}
	\label{Figure_03}
\end{figure}

Before introducing the multi-stage sensing strategy, we first recall the key assumptions:

\begin{itemize}
	\item 
	The azimuth angles of the targets and the scatterers, i.e., $\theta_m$ and $\eta_k$, are assumed to be constant throughout the entire sensing duration across the multiple stages.
	
	\item
	The pathloss coefficients of the targets and the scatterers, i.e., $\alpha_m$ and $\beta_k$, are assumed to be independent random variables distributed as $\mathcal{CN} ( 0, \varsigma_m^2 )$ and $\mathcal{CN} ( 0,\sigma_k^2 )$, respectively.
	Their second-order statistics $\varsigma_m^2$ and $\sigma_k^2$ are assumed to be known a priori.
	Because we do not track the values of $\alpha_m$ and $\beta_k$, we use this same prior for all the sensing stages\footnote{
		Better performance can be obtained if the posterior distributions of $\alpha_m$ and $\beta_k$ are also tracked across the stages, but at significant complexity cost, not only in computing the posterior distribution, but also in subsequent optimization, because the posterior of $\alpha_m$ and $\beta_k$ would typically not have zero means.
	}.
\end{itemize}
The above assumptions form the basis of the proposed sensing strategy presented below.
Specifically, the entire sensing duration is divided into multiple sensing stages, each consisting of $N$ symbol periods,
as illustrated in Fig.~\ref{Figure_03}. 
Then, the following operations are performed:
\begin{itemize}
	\item 
	The precoders (and combiners,~if present) for each stage are optimized based on the prior from the previous stage.
	The optimization utilizes the single-stage design algorithm presented previously.
	\item 
	At the end of each stage, the posterior of the target angles $\theta_m$ is inferred from the echoes collected within $N$ symbol periods, and then adopted as the prior for the next stage.
\end{itemize}

\newcounter{storeeqcounter}

\begin{figure*}[!t]
	\setcounter{storeeqcounter}{\value{equation}} 
	\begin{align}
	\setcounter{equation}{\value{storeeqcounter}+3} 
		\tilde{q}
		\left( \boldsymbol{\theta} \left| \mathbf{y}^{[i]} \right. \right)
		& \propto
		\tilde{q}
		\left( \boldsymbol{\theta} \left| \mathbf{y}^{[i-1]} \right. \right)
		\mathop{\cdot}
		\mathcal{CN} \left( \mathbf{y}^{[i]} \left| \mathbf{0}, \mathbf{\Sigma} \left( \mathbf{x}^{[i]} \right) + \sum_{m=1}^{M} p \varsigma_m^2 \left( \mathbf{U}_m \mathbf{x}^{[i]} \right) \left( \mathbf{U}_m \mathbf{x}^{[i]} \right)^{\mathsf{H}} \right. \right)  \label{posterior_d}  \\
		\tilde{q} \left( \boldsymbol{\theta} \left| \mathbf{r}^{[i]} \right. \right)
		& 
		\propto
		\tilde{q} \left( \boldsymbol{\theta} \left| \mathbf{r}^{[i-1]} \,\right. \right)
		\mathop{\cdot}\,
		\mathcal{CN} \left( \mathbf{r}^{[i]} \left| \hspace{1pt} \mathbf{0}, \mathbf{\Sigma} \left( \mathbf{x}^{[i]}, \mathbf{w}^{[i]} \right)
		+ \sum_{m=1}^{M} p \varsigma_m^2 \left( \mathbf{\bar{W}}^{[i]} \mathbf{U}_m \mathbf{x}^{[i]} \right) \left( \mathbf{\bar{W}}^{[i]} \mathbf{U}_m \mathbf{x}^{[i]} \right)^{\mathsf{H}} \right. \right) \label{posterior_a}
	\end{align}
	\setcounter{equation}{\value{storeeqcounter}} 
	\hrule
\end{figure*}

We now present the formulas for the iterative updates of the posterior.
As analyzed in Section~\ref{Section_03}, to make the likelihood function tractable, the aggregate interference-plus-noise term is modelled as a Gaussian random vector with matching mean and covariance.
For the fully digital receiver case, the posterior of the target parameters $\mathbf{v}$ at the $i$-th stage is expressed as
\begin{equation} \label{jointPos}
	\tilde{q} \big( \mathbf{v} \mathop{\big|} \mathbf{y}^{[i]} \big)
	\propto
	q ( \mathbf{v} )
	\mathop{\cdot}
	\mathcal{L} \big( \mathbf{y}^{[i]} \mathop{\big|} \mathbf{v} \big),
\end{equation}
where $\mathbf{y}^{[i]}$ denotes the echo signals received at the $i$-th stage, 
and $\mathcal{L} ( \mathbf{y}^{[i]} \mathop{|} \mathbf{v} )$ denotes the likelihood function of $\mathbf{v}$ given $\mathbf{y}^{[i]}$, as specified in \eqref{likelihood}.
Since $\{ \alpha_m \}_{m=1}^{M}$ are independent random variables 
and $\alpha_m \sim \mathcal{CN} (0, \varsigma_m^2)$, the prior $q(\mathbf{v})$ is given by
\begin{equation}
	q(\mathbf{v})
	=
	\tilde{q} \big( \boldsymbol{\theta} \mathop{\big|} \mathbf{y}^{[i-1]} \big)
	\mathop{\cdot}
	\prod_{m=1}^{M}
	\mathcal{CN} 
	\left( \left. \alpha_m \right| 0, \varsigma_m^2 \right),
\end{equation}
where $\tilde{q}(\boldsymbol{\theta} \mathop{|} \mathbf{y}^{[i-1]} )$ denotes the posterior of $\boldsymbol{\theta}$ from the previous stage and is used as the prior for the current stage.
Then, the posterior of $\mathbf{v}$ can be rewritten as
\begin{multline}
	\tilde{q} \big(\mathbf{v} \mathop{\big|} \mathbf{y}^{[i]} \big)
	\propto  \\
	\tilde{q} \big( \boldsymbol{\theta} \mathop{\big|} \mathbf{y}^{[i-1]} \big)
	\mathop{\cdot}
	\prod_{m=1}^{M}
	\mathcal{CN} \left( \left. \alpha_m \right| 0, \varsigma_m^2 \right)
	\mathop{\cdot}
	\mathcal{L} \big( \mathbf{y}^{[i]} \mathop{\big|} \mathbf{v} \big).
\end{multline}
Marginalizing the above posterior over $\boldsymbol{\alpha}$ leads to the posterior of the target angle $\boldsymbol{\theta}$ in \eqref{posterior_d} on the top of the page.
Following similar derivations as in the fully digital receiver case, the posterior of $\boldsymbol{\theta}$ for the single-RF-chain receiver can be obtained in \eqref{posterior_a} on the top of the page.

\addtocounter{equation}{2} 

As the posteriors are updated across the stages, the proposed multi-stage sensing strategy progressively refines the distributions of the target angles over successive sensing stages.

\section{Numerical Results}
\label{Section_08}

In this section, we present simulation results to illustrate the effectiveness and the efficiency of the proposed BCRLB based optimization\footnote{In Appendix~\ref{app:BCRB_vs_MSE}, we provide numerical results to show that the BCRLB is a good approximation to the actual MSE of the maximum a posterior estimator for reasonable SNRs and prior distributions in the angle estimation problem.}.
The simulation environment is set as follows:
\begin{itemize}
	
\item The normalized transmit power $p \varsigma_m^2 / \varepsilon^2$ is set to $-10$ dB.
The normalized interference power $p \sigma_k^2 / \varepsilon^2$ is set to $10$ dB.
		
\item
The numbers of transmit and receive antennas are set to be $N_T = N_R = 8$.
The spacing between adjacent antennas is set to half a wavelength.
	
\item 
A single sensing stage consists of $N = 10$ symbol periods.

\item The initial prior distributions of the target azimuth angles are chosen to be approximately uniform with a smooth raised-cosine taper of smoothing width $\zeta = 0.5^\circ$.
This is to ensure that the BFIM for the prior is well defined. The details of the approximation are described in Appendix~\ref{app:smooth_uniform_prior}. In the rest of the paper, we simply refer to such approximation as the uniform distribution. 
\item 
We use MOSEK \cite{mosek2019} to solve the convex optimization problem in each iteration.
	
\end{itemize}

\subsection{MIMO Sensing Beampatterns in Single Stage}

We use the beampatterns to illustrate the effectiveness of the proposed algorithm and provide insights into the resulting solutions.
The transmit beampattern of the precoder at the $n$-th symbol period is defined as
\begin{align}
	\hat{S}_{n} (\theta)
	\triangleq
	\left| \mathbf{h}_T^{\mathsf{T}}(\theta) \mathbf{x}_n \right|^2, 
	\quad
	n = 1, 2, \ldots, N.
\end{align}
Similarly, the receive beampattern of the combiner at the $n$-th symbol period is defined as
\begin{align}
	\check{S}_{n} (\theta)
	\triangleq
	\left|  \mathbf{w}_n^{\mathsf{T}} \mathbf{h}_R(\theta) \right|^2, 
	\quad
	n = 1, 2, \ldots, N.
\end{align}
Note that the transmit/receive beampatterns do not characterize the system's overall ability to suppress the interference.
To~show that the proposed solutions can generate nulls aligned with the directions of interfering scatterers, we further define the overall beampattern.
Using the digital-receiver case as an example,~the overall beampattern is defined as
\begin{equation}
	Q(\theta)
	=
	\left| 
	\mathbf{c}^{\mathsf{H}} 
	\left(  \mathbf{I}_N \otimes \mathbf{H}(\theta) \right) \mathbf{x}
	\right|^2,
\end{equation}
where $\mathbf{c}$ denotes a linear minimum mean
square error (LMMSE) combiner. 
Note that in practice, the sensing operation is based on the received vector signal and does not require such a combiner.
We use the above definition of beampattern to illustrate that the obtained solutions have intuitive interpretations.

%

In Fig.~\ref{Figure_05}, we present the beampatterns for the case of phase-only transmitter with digital receiver for a scenario in which
two targets are assumed to have uniform distributions in $[20^{\circ}, 60^{\circ}]$ and $[80^{\circ}, 120^{\circ}]$, respectively,
while one interfering scatterer is assumed to be uniformly distributed in $[135^{\circ}, 165^{\circ}]$.
The first three subfigures show the transmit beampatterns obtained over $N = 1$, $4$, and $10$ symbol periods, respectively.
The plots reveal that the obtained solutions attempt to cover the entire target regions.
However, because of the limited DoF of phase-only precoders, the beampattern with $N = 1$ cannot fully illuminate the entire target regions.
As $N$ increases, the resulting beams exhibit a \emph{sweeping-like} pattern, enabling the ensemble of beams to effectively cover the entire target regions.
The last subfigure presents the overall beampattern.
One can observe that the overall beampattern can generate beams aligned with all the target regions, with nulls well formed at the interference directions.
These aligned peaks and nulls help focus power on the desired
angles while reducing interference, thereby improving the sensing
performance.

\begin{figure}[!t]
	\centering
	\includegraphics[width = 0.965 \linewidth]{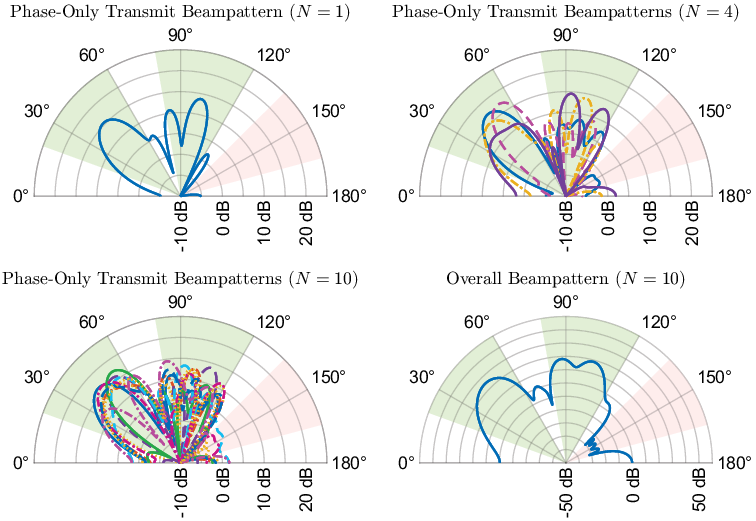}
	\caption{Phase-only transmit beampatterns with different numbers of symbol periods for the two-target case.}
	\label{Figure_05}
\end{figure}

\begin{figure}[!t]
	\centering
	\includegraphics[width = 0.947 \linewidth]{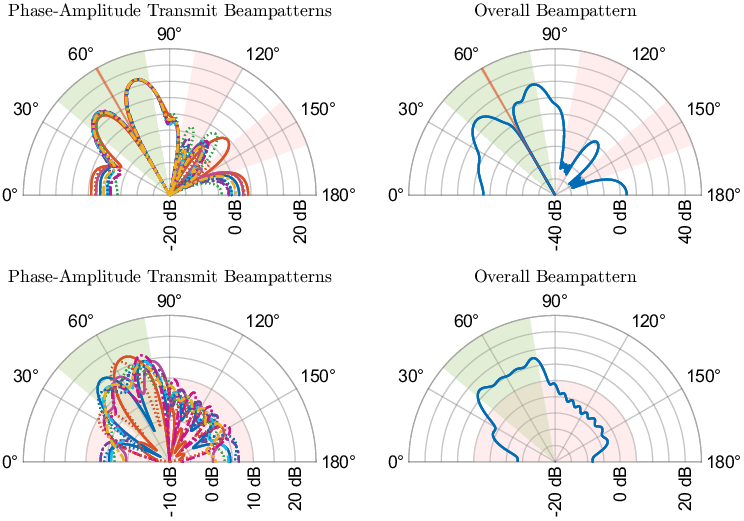}
	\caption{Phase-amplitude transmit beampatterns and overall beampatterns.}
	\label{Figure_08}
\end{figure}

In Fig.~\ref{Figure_08}, we show the phase-amplitude transmit beampatterns and the corresponding overall beampatterns under two different priors of scatterers, again with digital receiver.
In the top row, three interfering scatterers are considered: one at $60^\circ$ inside the target region $[40^\circ,80^\circ]$, and two uniformly distributed in $[100^\circ,120^\circ]$ and $[140^\circ,160^\circ]$, respectively.
In the bottom row, a single interfering scatterer is uniformly distributed over the entire angular region $[0^\circ,180^\circ]$.
Similar to the results of phase-only precoders, the optimized phase-amplitude precoders vary across symbol periods to better facilitate target illumination and interference suppression in the respective spatial directions. 
For the top row, a particularly interesting observation is that the beampatterns 
produce a null at the interfering direction $60^\circ$ while maintaining radiation pattern over the remaining target region.
For the bottom row, since the prior of the scatterer is over the entire angular range, the overall beampattern attempts to concentrate power on the target while keeping low power levels elsewhere.

\begin{figure}[!t]
	\centering
	\includegraphics[width = 0.95\linewidth]{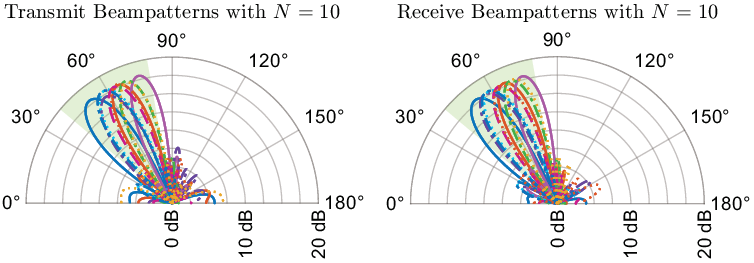}
	\caption{Phase-only transmit and receive beampatterns for a single target.}
	\label{Figure_07}
\end{figure}

Fig.~\ref{Figure_07} presents the beampatterns for the case where both the transmitter and the receiver have phase-only architecture. 
We consider a single target located in the angular region $[40^{\circ}, 80^{\circ}]$, with a uniform prior,
while several interfering scatterers are also present.
Unlike the digital-receiver case, where only the precoders need to be optimized, the phase-only transceiver case requires the precoders and combiners to be jointly designed.
One can see that both the transmit and receive beampatterns show a sweeping-like behaviour across all symbol periods, covering the entire target region.
Moreover, in each individual period, 
the precoder and combiner are steered toward the same direction, ensuring that both transmission and reception concentrate their energy along that direction, thereby enhancing the sensing performance.


\subsection{Evolution of Posterior Distributions over Multiple Stages}

To further illustrate the effectiveness of the proposed framework, we present the evolution of the posterior distributions~of the target angles oover several sensing stages.
Here, we focus on the more challenging phase-only precoding case.

In Fig.~\ref{Figure_09}, we present the results of the scenario in which~the azimuth angle of a single target has a uniform distribution in the range $\left[ 40^{\circ}, 80^{\circ} \right]$, and three interfering scatterers are located in $[100^{\circ}, 120^{\circ}]$, $[140^{\circ},160^{\circ}]$, and at $130^{\circ}$.
The true value of the target azimuth angle is set as $45^{\circ}$, 
and the parameters $\eta_k$, $\alpha_m$, $\beta_k$, and $\mathbf{n}$ are generated from their distributions.
Fig.~\ref{Figure_09} plots the posterior distributions of the target angle after three stages using the precoders (and combiners) designed by the proposed algorithm for both digital and phase-only cases.
One can observe that as the number of stages increases,
the posterior rapidly converges to a distribution with a peak at the actual sensing angle.
This result validates the effectiveness of the proposed algorithms.
Moreover, compared to digital-receiver case,
the convergence of the phase-only case is slower.
This result reveals the trade-off between low-cost hardware and estimation performance.

In Fig.~\ref{Figure_10}, we present the results for the two-target scenario. 
For clarity, we only present the digital-receiver results, as~the analog-receiver results are similar.
The azimuth angles $\theta_1$ and $\theta_2$ of the two targets follow Gaussian distributions with means at $45^\circ$ and $90^\circ$, respectively, and a standard deviation of $10^\circ$.
An interference scatterer is located in $[120^{\circ}, 150^{\circ}]$.
It can~be seen from Fig.~\ref{Figure_10} that the joint posterior distributions rapidly converge to a distribution with a peak around the true values.
This validates that the proposed algorithm is also effective for multi-target scenarios.

\begin{figure}[!t]
	\centering
	\includegraphics[width = 0.975\linewidth]{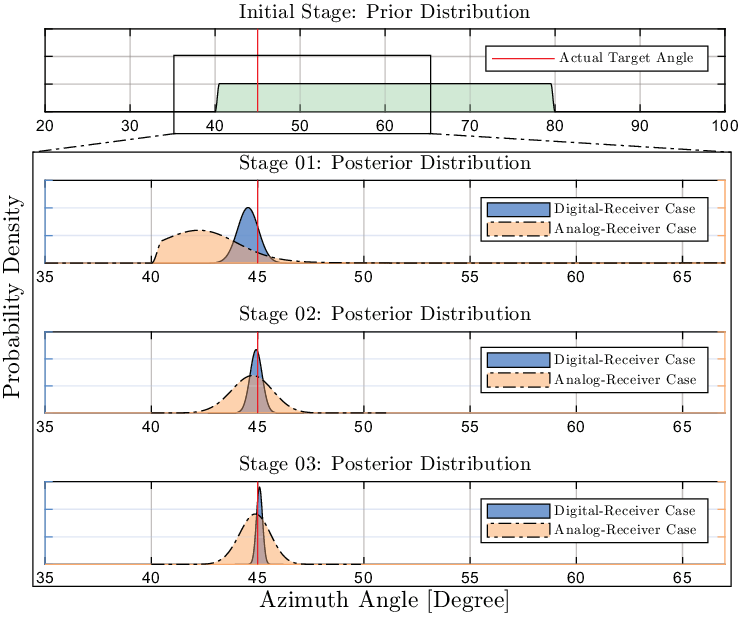}
	\caption{Posterior distributions of the target angle $\theta_1$ over three iterations with adaptive precoders and combiners.}
	\label{Figure_09}
\end{figure}

\begin{figure}[!t]
	\centering
	\includegraphics[width = 0.957\linewidth]{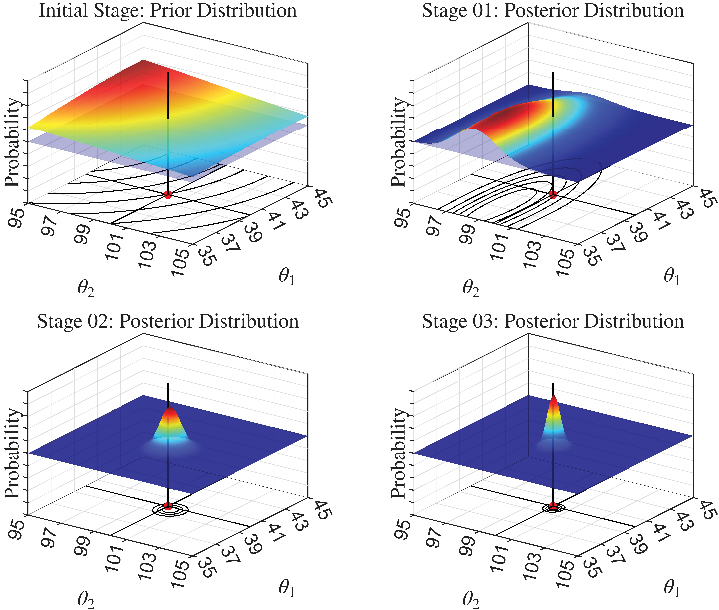}
	\caption{Posterior distributions of the two target angles $\theta_1$ and $\theta_2$ over three iterations with adaptive precoders.}
	\label{Figure_10}
\end{figure}

The preceding results are obtained by fixing the number of symbol periods per stage to $N = 10$.
Next, we show the effect of stage partitioning under a fixed total symbol-period budget.
Specifically, under the same settings as in Fig.~\ref{Figure_09} and using the digital-receiver case as an example, we evaluate the effect of allocating different number of symbol periods per stage. 
We use the standard deviation of the posterior to quantify the remaining uncertainty of the target angle.
One can observe from the right panel of Fig.~\ref{Figure_11} that adopting more stages leads to better sensing performance.
This is because more frequent posterior updates enable the precoders to focus resource on the most likely directions early on.
However, one can observe from the left panel of Fig.~\ref{Figure_11} that when the SNR is low, it is better to choose a larger $N$ per stage, which allows sweeping the angular range and obtains a better posterior.
Overall, 
there is a trade-off between the number of symbols per stage and the number of stages, balancing exploration and exploitation.

\begin{figure}[!t]
	\centering
	\includegraphics[width = 1\linewidth]{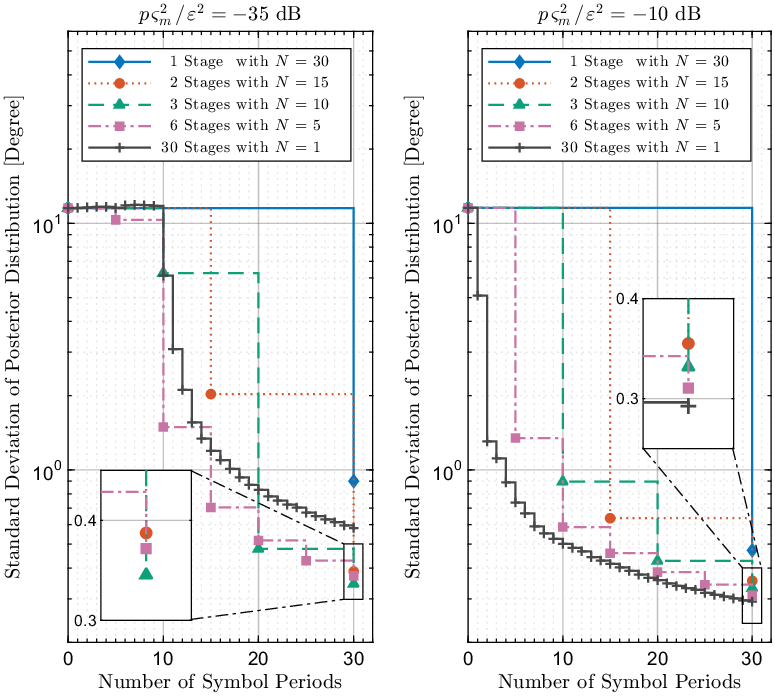}
	\caption{Evolution of the standard deviation of the posterior distribution across the stages for different number of symbols per stage.}
	\label{Figure_11}
\end{figure}

\subsection{Convergence Behaviour of the Optimization Algorithms}

We now show the convergence behaviour of the proposed optimization algorithms for the maximization and the max-min problems, corresponding to the single-target and the multi-target cases, respectively.

We adopt the coordinate ascent (CA) as a baseline. 
In this benchmark, we first quantize the phase shifts of the phase-only precoders into $8$-bit levels, 
which provides a sufficiently fine resolution without a noticeable performance degradation.
Then, we use the CA to alternately optimize each coordinate.
For the maximization problem, the projected gradient ascent (PGA) is adopted as another benchmark.
We set the gradient to be normalized, and the step size to a fixed value of $0.1$.

\begin{figure}[!t]
	\centering
	\includegraphics[width = 1\linewidth]{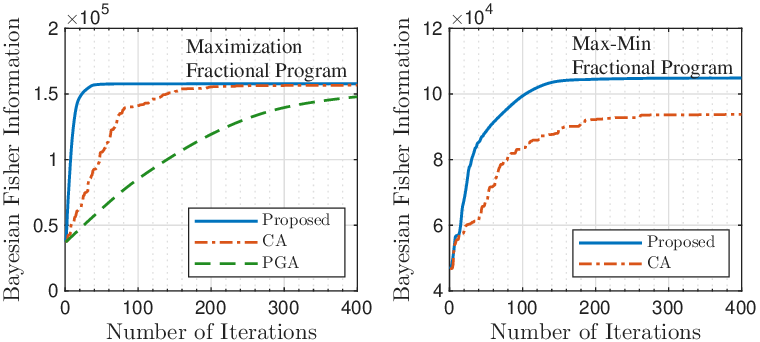}
	\caption{Convergence curves of different algorithms for the maximization and max–min fractional programming problems.}
	\label{Figure_13}
\end{figure}

\begin{table}[!t]
	\centering
	\caption{Typical Runtime of Different Algorithms}
	\renewcommand{\arraystretch}{1}
	\label{TableI}
	\begin{tabular}{l|c c c}
		\toprule
		{Algorithm} 
		& {Runtime/Iteration\;[s]} 
		& {No.\;Iterations} 
		& {Total Runtime\;[s]}  \\ 
		\midrule
		\multicolumn{4}{l}{\hspace{46pt}\textit{Maximization Fractional Program with $N_T N = 80$}}  \\ 
		\midrule
		{Proposed}    & 0.0072  & 51   & 0.3663   \\
		{CA}          & 0.1762  & 237  & 41.7582  \\ 
		{PGA}         & 0.2304  & 813  & 187.3403 \\
		\midrule
		\multicolumn{4}{l}{\hspace{46pt}\textit{Max-Min Fractional Program with $N_T N = 80$}}  \\ 
		\midrule
		{Proposed}    & 0.0077  & 187  & 1.4455  \\
		{CA}          & 0.0671  & 239  & 16.0369 \\  
		\bottomrule
	\end{tabular}
\end{table}

In Fig.~\ref{Figure_13} and Table~\ref{TableI}, we show the convergence curves and typical runtimes of the proposed algorithms and the benchmarks for the maximization and max-min problems, corresponding to the single-target and multi-target cases, respectively. 
It can be observed that the proposed algorithms achieve the best performance while requiring the least amount of runtime. 

\section{Conclusion}
\label{Section_09}

This paper shows that the prior information can be used to separate targets from scatterers in a MIMO sensing framework. We develop sequence-level precoder and combiner design techniques for MIMO sensing with the scatterers modelled as signal-dependent interference. 
We formulate the problem of worst-case BCRLB minimization among all the target angles under both the phase-only and phase-amplitude precoder and combiner architectures. 
A constant-norm extension of the linear transform is developed to handle the resulting max-min fractional program, leading to an efficient iterative algorithm involving low-dimensional convex optimization and closed-form precoder/combiner updates. 
Numerical results demonstrate that the proposed design can synthesize effective sweeping-like beampatterns for target illumination and scatterer suppression, and reveal a trade-off between the number of sensing stages and the number of symbols allocated to each stage.

\appendices

\section{Proof of Lemma~\ref{Lemma_1}}
\label{app: Lemma_01}

To establish the constant-norm linear transform for $f(\mathbf{x})$, we first apply the quadratic transform technique \cite{8314727} to \eqref{def}. 
Specifically, a lower bound for $f(\mathbf{x})$ is given by
\begin{align} \label{qua_trans}
	f(\mathbf{x})
	\geq 
	2 \mathop{\mathfrak{Re}} \left\lbrace \left( \mathbf{A} \mathbf{x} \right)^{\mathsf{H}} \mathbf{g} \right\rbrace
	- 
	\mathbf{g}^{\mathsf{H}} \mathbf{D} (\mathbf{x}) \mathbf{g}, \;\; \forall \, \mathbf{x}, \mathbf{g},
\end{align}
with the equality achieved at
\begin{equation}
	\mathbf{g}
	= 
	\left( \mathbf{D} (\mathbf{x}) \right)^{-1} \mathbf{A} \mathbf{x}.
\end{equation}
The quadratic term in \eqref{qua_trans} can be written as 
\begin{align} \label{13}
	\mathbf{g}^{\mathsf{H}} \mathbf{D} (\mathbf{x}) \mathbf{g}
	= 
	\mathbf{x}^{\mathsf{H}} \mathbf{M} \mathbf{x}
	+
	\mathbf{g}^{\mathsf{H}} \mathbf{C} \mathbf{g},
\end{align}
where the matrix $\mathbf{M}$ is given by
\begin{equation}
	\mathbf{M}
	=  
	\sum_{k} \rho_k
	\left( \mathbf{B}_{k}^{\mathsf{H}} \mathbf{g} \right) 
	\left( \mathbf{B}_{k}^{\mathsf{H}} \mathbf{g} \right)^{\mathsf{H}}.
\end{equation}	
We now eliminate the quadratic term $\mathbf{x}^{\mathsf{H}} \mathbf{M} \mathbf{x}$ in (\ref{13}) using a majorization-minimization technique \cite[Eq.~(26)]{7547360},
\begin{align} \label{18}
	\mathbf{x}^{\mathsf{H}} \mathbf{M} \mathbf{x}
	\leq \mathbf{x}^{\mathsf{H}} \mathbf{L} \mathbf{x}
	& 
	+ \mathbf{z}^{\mathsf{H}} \left( \mathbf{L} - \mathbf{M}  \right) \mathbf{z}  \notag \\
	& 
	+ 2 \mathop{\mathfrak{Re}} \left\lbrace
	\mathbf{x}^{\mathsf{H}} 
	\left( \mathbf{M} - \mathbf{L} \right) \mathbf{z} \right\rbrace,
\end{align}
where $\mathbf{L} \succeq \mathbf{M}$, and the equality is achieved at $\mathbf{z} = \mathbf{x}$.
Then, by~replacing $\mathbf{L}$ with $\delta \mathbf{I}$, where $\delta$ denotes the trace of $\mathbf{M}$ so that $\delta \mathbf{I} \succeq \mathbf{M}$, and by combining with 
the fact that for any $\mathbf{x}$ and $\mathbf{z}$ in the \emph{constant-norm} set $\mathcal{X}
= \{ \mathbf{x} \mathop{|} \|\mathbf{x}\|_2 = \sqrt{N_T N}\}$,
\begin{equation} \label{10}
	\mathbf{x}^{\mathsf{H}} \left( \delta \mathbf{I} \right) \mathbf{x} 
	= \mathbf{z}^{\mathsf{H}} \left( \delta \mathbf{I} \right) \mathbf{z} 
	= \delta N_T N ,
\end{equation}
we obtain 
\begin{multline}
	f(\mathbf{x}) 
	\geq 
	2 \mathop{\mathfrak{Re}} \left\lbrace \mathbf{x}^{\mathsf{H}} \left( \left( \delta \mathbf{I} - \mathbf{M} \right) \mathbf{z} 
	+ 
	\mathbf{A}^{\mathsf{H}} \mathbf{g} \right)  \right\rbrace  \\
	+ 
	\mathbf{z}^{\mathsf{H}} \mathbf{M} \mathbf{z} 
	- 
	2 \delta N_T N
	- 
	\mathbf{g}^{\mathsf{H}} \mathbf{C} \mathbf{g} .
\end{multline}
The equality is achieved when $\mathbf{z} = \mathbf{x}$ and $\mathbf{g} = \left( \mathbf{D} (\mathbf{x}) \right)^{-1} \mathbf{A} \mathbf{x}$.

\section{Proof of Theorem~\ref{theorem_1}}
\label{app: Theorem_01}

By introducing an auxiliary variable $\xi$, the problem~\eqref{newmax-minLP} can be equivalently rewritten as
\begin{subequations}  \label{107}
	\begin{align}
		\mathop{\mathrm{maximize}}_{\substack{\mathbf{x}, \mathbf{z} \\ \{\mathbf{G}_{\mathsf{o},m}\}, \{\mathbf{G}_{\mathsf{c},w}\}}}\;\;
		&\;
		\xi  \\
		\mathop{\mathrm{subject\;\;to}} \quad\,
		&\;
		h_{\mathsf{o},m} \big(\bar{\mathbf{f}}_{\mathsf{o},m}(\mathbf{x}, \mathbf{z}, \mathbf{G}_{\mathsf{o},m})\big)
		\geq \xi, 
		\ \ \forall m,  \label{ineq_1}  \\
		&\;
		h_{\mathsf{c},w}
		\big(\bar{\mathbf{f}}_{\mathsf{c},w} (\mathbf{x}, \mathbf{z}, \mathbf{G}_{\mathsf{c},w}) \big) 
		\geq \Gamma_w, \ \forall w,  \label{ineq_2} \\
		&\;
		\mathbf{x} \in \mathcal{X}, 
	\end{align}
\end{subequations}
where 
\begin{align}
	& \bar{\mathbf{f}}_{\mathsf{o},m}(\mathbf{x}, \mathbf{z}, \mathbf{G}_{\mathsf{o},m})
	=  \notag  \\
	& \qquad \quad 
	\left[ \bar{f}_{\mathsf{o},m}^{[1]}(\mathbf{x}, \mathbf{z}, \mathbf{g}_{\mathsf{o},m}^{[1]}), \ldots, \bar{f}_{\mathsf{o},m}^{[R_m]}(\mathbf{x}, \mathbf{z}, \mathbf{g}_{\mathsf{o},m}^{[R_m]}) \right],  \\
	& \bar{\mathbf{f}}_{\mathsf{c},w}(\mathbf{x}, \mathbf{z}, \mathbf{G}_{\mathsf{c},w})
	=  \notag  \\
	& \qquad \quad 
	\left[ \bar{f}_{\mathsf{c},w}^{[1]}(\mathbf{x}, \mathbf{z}, \mathbf{g}_{\mathsf{c},w}^{[1]}), \ldots, \bar{f}_{\mathsf{c},w}^{[T_w]}(\mathbf{x}, \mathbf{z}, \mathbf{g}_{\mathsf{c},w}^{[T_w]}) \right].
\end{align}

First, we show that the above problem is equivalent to
\begin{subequations} \label{108}
	\begin{align}
		\mathop{\mathrm{maximize}}_{\xi, \mathbf{x}} \;\;
		&\;
		\xi  \\
		\mathop{\mathrm{subject \; to}} \;\;
		&\;
		h_{\mathsf{o},m} \big(\hat{\mathbf{f}}_{\mathsf{o},m}(\mathbf{x})\big)
		\geq \xi, 
		\ \ \forall m,  \label{ineq_01}  \\
		&\;
		h_{\mathsf{c},w}
		\big(\hat{\mathbf{f}}_{\mathsf{c},w} (\mathbf{x})\big) 
		\geq \Gamma_w, \ \forall w,  \label{ineq_02}  \\
		&\;
		\mathbf{x} \in \mathcal{X}, 
	\end{align}
\end{subequations}
where
\begin{align}
	& \hat{\mathbf{f}}_{\mathsf{o},m}(\mathbf{x})
	\triangleq  \notag  
	\\
	& \quad 
	\left[ 
	\max_{\mathbf{z}, \mathbf{g}_{\mathsf{o},m}^{[1]}}
	\bar{f}_{\mathsf{o},m}^{[1]}(\mathbf{x}, \mathbf{z}, \mathbf{g}_{\mathsf{o},m}^{[1]}), 
	\ldots, 
	\max_{\mathbf{z}, \mathbf{g}_{\mathsf{o},m}^{[R_m]}} 
	\bar{f}_{\mathsf{o},m}^{[R_m]}(\mathbf{x}, \mathbf{z}, \mathbf{g}_{\mathsf{o},m}^{[R_m]}) 
	\right],  \label{f_hat_o}
	\\
	& \hat{\mathbf{f}}_{\mathsf{c},w}(\mathbf{x})
	\triangleq  \notag  
	\\
	& \quad 
	\left[ 
	\max_{\mathbf{z}, \mathbf{g}_{\mathsf{c},w}^{[1]}}
	\bar{f}_{\mathsf{c},w}^{[1]}(\mathbf{x}, \mathbf{z}, \mathbf{g}_{\mathsf{c},w}^{[1]}), 
	\ldots, 
	\max_{\mathbf{z}, \mathbf{g}_{\mathsf{c},w}^{[T_w]}}
	\bar{f}_{\mathsf{c},w}^{[T_w]}(\mathbf{x}, \mathbf{z}, \mathbf{g}_{\mathsf{c},w}^{[T_w]}) 
	\right]. \label{f_hat_c}
\end{align}
To show the above equivalence, we show that the feasible set of $\mathbf{x}$ satisfying \eqref{ineq_01} and \eqref{ineq_02}
is the same as that satisfying \eqref{ineq_1} and \eqref{ineq_2}.

Let $\mathbf{x}$ be a point satisfying the constraints \eqref{ineq_1} and \eqref{ineq_2}. 
Then, $\mathbf{x}$ must satisfy the constraints \eqref{ineq_01} and \eqref{ineq_02}.
This is because the latter constraints are defined in terms of maximizations of $\bar{f}_{\mathsf{o},m}^{[r_m]}(\mathbf{x}, \mathbf{z}, \mathbf{g}_{\mathsf{o},m}^{[r_m]})$ and $\bar{f}_{\mathsf{c},w}^{[t_w]}(\mathbf{x}, \mathbf{z}, \mathbf{g}_{\mathsf{c},w}^{[t_w]})$ 
over $\mathbf{z}$, $\mathbf{g}_{\mathsf{o},m}^{[r_m]}$ and $\mathbf{g}_{\mathsf{c},w}^{[t_w]}$.
For $\mathbf{g}_{\mathsf{o},m}^{[r_m]}$ and $\mathbf{g}_{\mathsf{c},w}^{[t_w]}$, they
appear separately in the constraints \eqref{f_hat_o} and \eqref{f_hat_c}. 
For the variable $\mathbf{z}$, although it appears in both  \eqref{f_hat_o} and \eqref{f_hat_c}, for fixed $\mathbf{x}$, the \emph{same} value of $\mathbf{z}$ (namely $\mathbf{z}^{\star} = \mathbf{x}$) maximizes all the ratio functions in \eqref{f_hat_o} and \eqref{f_hat_c}. 
This crucial observation, together with the fact that
the outer functions $h_{\mathsf{o},m}(\cdot)$ and $h_{\mathsf{c},w}(\cdot)$ are non-decreasing in each component, establish the statement. 

Conversely, if $\mathbf{x}$ is a feasible point in the constraints \eqref{ineq_01} and \eqref{ineq_02},
$(\mathbf{x}, \mathbf{z}^{\star}, (\mathbf{g}_{\mathsf{o},m}^{[r_m]})^{\star}, (\mathbf{g}_{\mathsf{c},w}^{[t_w]})^{\star})$ is a feasible point in the constraints \eqref{ineq_1} and \eqref{ineq_2}.
Hence, the feasible set of $\mathbf{x}$ in the problem \eqref{107} is the same as that in the problem \eqref{108}.
This shows that the two problems are equivalent.

Second, we show the equivalence between the problems \eqref{108}
and \eqref{max-min-ratios}.
Based on Lemma~\ref{Lemma_1}, we have
\begin{align}
	\max_{\mathbf{z}, \mathbf{g}_{\mathsf{o},m}^{[r_m]}}
	\bar{f}_{\mathsf{o},m}^{[r_m]}(\mathbf{x}, \mathbf{z}, \mathbf{g}_{\mathsf{o},m}^{[r_m]})
	&=  
	f_{\mathsf{o},m}^{[r_m]}(\mathbf{x}),  \\
	\max_{\mathbf{z}, \mathbf{g}_{\mathsf{c},w}^{[t_w]}}
	\bar{f}_{\mathsf{c},w}^{[t_w]}(\mathbf{x}, \mathbf{z}, \mathbf{g}_{\mathsf{c},w}^{[t_w]})
	&=
	f_{\mathsf{c},w}^{[t_w]}(\mathbf{x}).
\end{align}
Therefore, the problem \eqref{108} is equivalent to
\begin{subequations}
	\begin{align}
		\mathop{\mathrm{maximize}}_{\xi, \mathbf{x}} \;\;
		&\;
		\xi  \\
		\mathop{\mathrm{subject \; to}} \;\;
		&\;
		h_{\mathsf{o},m}
		\big( \mathbf{f}_{\mathsf{o},m}(\mathbf{x}) \big) 
		\geq \xi, 
		\;\;\; \forall m, \\
		&\;
		h_{\mathsf{c},w}
		\big( \mathbf{f}_{\mathsf{c},w} (\mathbf{x})\big) 
		\geq \Gamma_w, \ \forall w,  \\
		&\;
		\mathbf{x} \in \mathcal{X}, 
	\end{align}
\end{subequations}
where
\begin{align}
	\mathbf{f}_{\mathsf{o},m}(\mathbf{x})
	& \triangleq	
	\left[ f_{\mathsf{o},m}^{[1]}(\mathbf{x}), 
	\ldots, 
	f_{\mathsf{o},m}^{[R_m]}(\mathbf{x}) \right],  
	\\
	\mathbf{f}_{\mathsf{c},w}(\mathbf{x})
	& \triangleq	
	\left[ f_{\mathsf{c},w}^{[1]}(\mathbf{x}), 
	\ldots, 
	f_{\mathsf{c},w}^{[T_w]}(\mathbf{x}) \right].
\end{align}
Therefore, the problem \eqref{newmax-minLP} is equivalent to  \eqref{max-min-ratios}.


\section{BCRLB and MSE for Azimuth Angle Estimation}
\label{app:BCRB_vs_MSE}

Unlike classical CRB, the BCRLB is not attainable even asymptotically 
\cite{10584423} for general parameter estimation problems. 
In this appendix, we provide numerical results to show that 
for the specific azimuth angle estimation problem 
in this paper, the BCRLB can be a good approximation of 
the MSE for a maximum a posteriori (MAP) estimator, thus justifying
its use for beampattern optimization.
The simulation is carried out under the same sensing model and design 
assumptions as those used in the main text.
Specifically, the precoding sequence is designed using the proposed 
framework. We consider the single-target estimation problem for both the case of known pathloss coefficient
and the case of unknown pathloss coefficient with a zero-mean prior
distribution. 
The results show that, under the considered model assumptions, the MSE of the MAP estimator closely follows the BCRLB in the moderate-to-high receive-SNR\footnote{For convenience, the SNR is defined to be $\frac{p |\alpha_m|^2}{\varepsilon^2}$ without accounting for the effect of precoding and combining.} regime, as long as the prior distribution of the angle is not too wide.
This supports the use of the BCRLB as a practical surrogate objective for the proposed Bayesian precoder design.

\begin{figure}[!t]
	\centering
	\includegraphics[width = 1\linewidth]{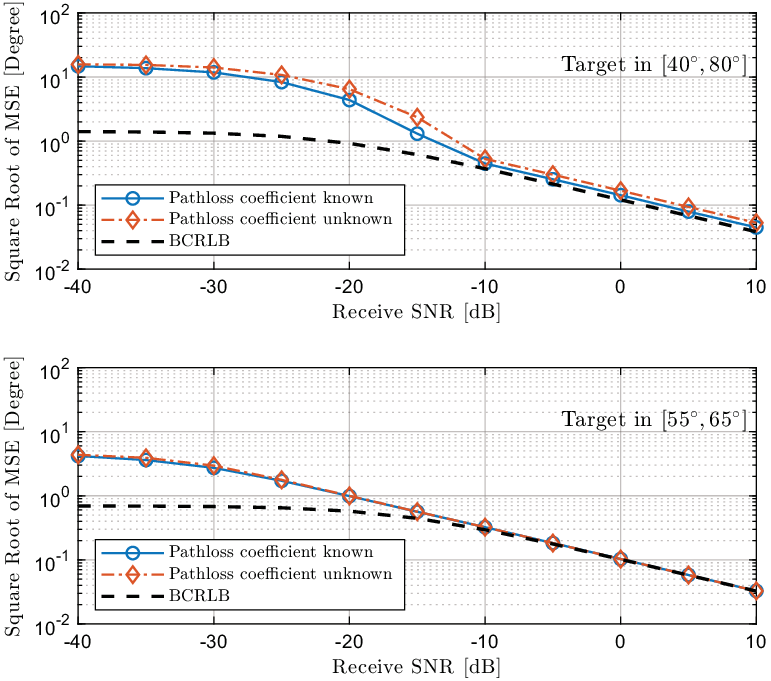}
	\caption{BCRLB and the square root of MSE for azimuth angle estimation with different priors.}
	\label{Figure_12}
\end{figure}

In Fig.~\ref{Figure_12}, we present the MSE of the MAP estimator for target-angle estimation and the corresponding BCRLB under two different target-angle priors.
In this example, a single target is assumed to be located in the angular region $[40^{\circ}, 80^{\circ}]$ or $[55^{\circ}, 65^{\circ}]$, with a uniform prior.
Moreover, three interfering scatterers are assumed to be present: 
two are uniformly distributed in $[100^{\circ}, 120^{\circ}]$ and $[140^{\circ}, 160^{\circ}]$, respectively, while the third is located at a deterministic angle of $130^{\circ}$.
The phase-only precoding sequence across $N=10$ symbol periods is optimized by the proposed framework.

It can be observed from the figure that under both prior distributions, 
the MSE decreases with the receive SNR in a manner consistent with the BCRLB 
and remains close to it in the moderate-to-high SNR regime.
Compared with the prior of $[55^\circ,65^\circ]$, the wider prior $[40^\circ,80^\circ]$ leads to a slightly larger gap from the BCRLB, but the overall there is excellent agreement between the BCRLB and MSE. 
This is true for both the case of known pathloss and the case of unknown pathloss with a zero-mean prior.


At the low-SNR regime, the observations become weakly informative. The estimation accuracy is then mainly governed by the prior information.
In this case, the BCRLB is not a tight approximation of the MSE.

\section{Smooth Approximation of Uniform Priors for BCRLB Computation}
\label{app:smooth_uniform_prior}

The definition of the Bayesian Fisher information in \eqref{prior_FI} requires certain regularity conditions on the prior which are not satisfied by the uniform distribution. 
This appendix shows that we can approximate the probability density function of the uniform distribution using a smoothing taper function. This would allow the Fisher information to be computed. 

The uniform distribution over $[a_m,b_m]$ as defined below 
\begin{equation}
	q_m(\theta_m)
	=
	\frac{1}{b_m-a_m}
	\mathbf{1}_{[a_m,b_m]}(\theta_m),
\end{equation}
has discontinuities at the boundary points \(a_m\) and $b_m$.
In order to satisfy the regularity condition,
we approximate the uniform distribution using a raised-cosine-taper-based smoothing method \cite{10584423}.
Let $\zeta>0$ denote the smoothing width with $0 < \zeta < (b_m-a_m)/{2}$.
We define
\begin{equation}
	h_\zeta(\theta)
	=
	\begin{cases}
		0,  & \theta \leq 0,  \\[1ex]
		\dfrac{1-\cos\left(\pi \theta/\zeta\right)}{2},
		& 0 < \theta < \zeta,  \\[1ex]
		1, & \theta \geq \zeta.
	\end{cases}
\end{equation}
with its derivative as given by
\begin{equation}
	h_\zeta^{\prime}(\theta)
	=
	\begin{cases}
		0, 
		& \hspace{10pt} \theta \leq 0,  \\[1ex]
		\dfrac{\pi \sin\left(\pi \theta/\zeta\right)}{2\zeta},
		& \hspace{10pt} 0 < \theta < \zeta,  \\[1ex]
		0,
		& \hspace{10pt} \theta \geq \zeta.
	\end{cases}
\end{equation}
Then, the uniform prior distribution over $[a_m,b_m]$ can be approximated as
\begin{equation} \label{smooth_f}
	q_{m,\zeta}(\theta_m)
	\triangleq
	\dfrac{h_\zeta(\theta_m-a_m) h_\zeta(b_m-\theta_m)}{\int h_\zeta(x-a_m) h_\zeta(b_m-x) \mathrm{d} x}.
\end{equation}
This approximation removes the discontinuities of the uniform prior while keeping it approximately flat over the interior of its support.
Fig.~\ref{Figure_app_01} presents an example of the raised-cosine approximations of the uniform prior on $[40^\circ,80^\circ]$ with different smoothing widths $\zeta$. 
It can be observed that a smaller $\zeta$ gives a closer approximation to the uniform prior.

\begin{figure}[!t]
	\centering
	\includegraphics[width = 1\linewidth]{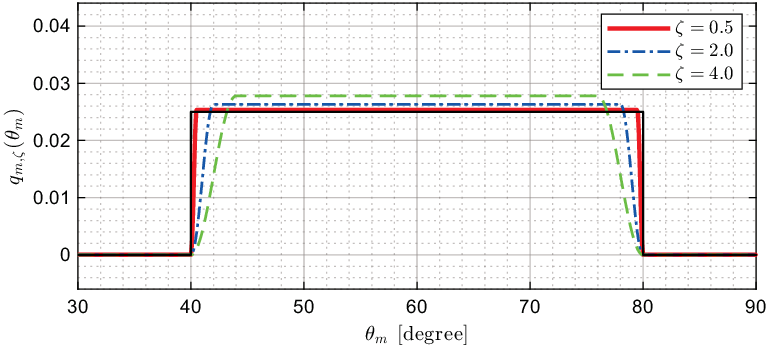}
	\caption{Smoothed uniform distributions on $[40^\circ, 80^\circ]$ with different $\zeta$.}
	\label{Figure_app_01}
\end{figure}

We now consider the Fisher information associated with the smoothed prior distribution $q_{m,\zeta}(\theta_m)$,
defined as the expectation of the squared score function, i.e.,
\begin{align} \label{104}
	\mathbf{\Phi}_{\mathsf{p},\zeta}^{(m,m)}
	&\triangleq
	\mathbb{E}_{q_{m,\zeta}(\theta_m)}
	\left[ \left( 
	\dfrac{\partial}{\partial\theta_m}
	\ln q_{m,\zeta}(\theta_m)
	\right)^2 \right].
\end{align}
Strictly speaking, the above definition requires $ q_{m,\zeta}(\theta_m)$ to be nonzero everywhere, which is not true outside of $(a_m,b_m)$. 
But we can always include a small probability density $\tilde\epsilon$ in the interval outside of $(a_m, b_m)$ to make \eqref{104} well defined. 
This would not contribute to the overall expectation. 


We now compute the Fisher information as in the integral below: 
\begin{equation} \label{integral_ab}
	\mathbf{\Phi}_{\mathsf{p},\zeta}^{(m,m)}
	=
	\int_{a_m}^{b_m}
	s_{m,\zeta}^2(\theta_m)
	q_{m,\zeta}(\theta_m)
	\mathrm{d} \theta_m.
\end{equation}
For $\theta_m \in (a_m, b_m)$, the score function of $q_{m,\zeta}(\theta_m)$ is given by
\begin{equation} \label{score}
	\frac{\partial}{\partial \theta_m}
	\ln q_{m,\zeta}(\theta_m)
	=
	\frac{
		h_\zeta'(\theta_m-a_m)
	}{
		h_\zeta(\theta_m-a_m)
	}
	-
	\frac{
		h_\zeta'(b_m-\theta_m)
	}{
		h_\zeta(b_m-\theta_m)
	}.
\end{equation}
For the first term in \eqref{score}, we have
\begin{multline}
	\frac{h_\zeta'(\theta_m-a_m)}{h_\zeta(\theta_m-a_m)}
	=  \\
	\begin{cases}
		\dfrac{\frac{\pi}{\zeta}\sin\left(\pi(\theta_m-a_m)/\zeta\right)}
		{1-\cos\left(\pi(\theta_m-a_m)/\zeta\right)},
		& a_m < \theta_m < a_m + \zeta,  
		\\[2ex]
		\hspace{1pt}0, & \theta_m \geq a_m + \zeta.
	\end{cases}
\end{multline}
For the second term in \eqref{score}, we have
\begin{multline}
	\frac{h_\zeta'(b_m-\theta_m)}{h_\zeta(b_m-\theta_m)}
	=  \\
	\begin{cases}
		\dfrac{\frac{\pi}{\zeta}\sin\left(\pi(b_m-\theta_m)/\zeta\right)}
		{1-\cos\left(\pi(b_m-\theta_m)/\zeta\right)},
		& b_m-\zeta < \theta_m < b_m,  
		\\[2ex]
		\hspace{1pt}0, & \theta_m \leq b_m-\zeta.
	\end{cases}
\end{multline}
The above two expressions show that, on any closed subinterval of $(a_m,b_m)$, the integrand in \eqref{integral_ab} is always bounded. 
We now examine the behaviour of the integrand when $\theta_m \downarrow a_m$ or $\theta_m \uparrow b_m$ and show that the integral is finite over $[a_m, b_m]$.

First consider the left boundary. 
Let
\begin{equation}
	\epsilon
	=
	\theta_m-a_m .
\end{equation}
As $\theta_m\downarrow a_m$, we have $\epsilon\to 0^+$. 
By using the Taylor expansion, we have
\begin{equation}
	h_\zeta(\epsilon)
	=
	\frac{\pi^2 \epsilon^2}{4\zeta^2}
	+
	O\left(\epsilon^4\right)
	=
	O\left(\epsilon^2\right),
\end{equation}
and
\begin{equation}
	h_\zeta'(\epsilon)
	=
	\frac{\pi^2 \epsilon}{2\zeta^2}
	+
	O\left(\epsilon^3\right)
	=
	O\left(\epsilon\right).
\end{equation}
Therefore, 
\begin{equation}
	\frac{
		h_\zeta'(\epsilon)
	}{
		h_\zeta(\epsilon)
	}
	=
	O\left(
	\frac{1}{\epsilon}
	\right).
\end{equation}

Moreover, 
for $\theta_m$ close to $a_m$ we have 
\begin{equation}
	h_\zeta(b_m-\theta_m)
	=
	1,
	\quad
	h_\zeta'(b_m-\theta_m)
	=
	0.
\end{equation}
Thus, near $a_m$, the second term in \eqref{score} is zero. The score function is therefore determined only by the first term.

Now, by the definition of $q_{m,\zeta}(\theta_m)$, we have 
\begin{equation}
	q_{m,\zeta}(\theta_m)
	=
	O\left(\epsilon^2\right).
\end{equation}
Since the score function is $O(1/\epsilon)$ near $a_m$ and
$q_{m,\zeta}(\theta_m)=O(\epsilon^2)$, we have
\begin{multline}
	\left(
	\frac{\partial}{\partial \theta_m}
	\ln q_{m,\zeta}(\theta_m)
	\right)^2
	q_{m,\zeta}(\theta_m)  \\
	=
	O\left(
	\frac{1}{\epsilon^2}
	\right)
	O\left(\epsilon^2\right)
	=
	O(1),
	\quad
	\theta_m\downarrow a_m .
\end{multline}
Therefore, the integrand in \eqref{integral_ab} is bounded in the right-neighborhood of $a_m$, and its integral over this neighborhood is finite.

The same argument applies to the right boundary $b_m$, with $b_m-\theta_m$ playing the role of $\epsilon$; therefore, the integrand is also bounded near $b_m$ and is integrable there. 
This shows that the integral in \eqref{integral_ab} is finite.

This allows us to compute the Fisher information via numerical integration as follows.
Define a grid over $(a_m,b_m)$:
\begin{equation}
	\theta_{m,i}
	=
	\theta_{m,\min}
	+
	(i-1)\Delta\theta_m,
	\quad
	i=1,\ldots,N_{\theta_m} .
\end{equation}
We compute the unnormalized prior values as
\begin{equation}
	\tilde{q}_{m,\zeta,i}
	=
	h_\zeta(\theta_{m,i}-a_m)
	h_\zeta(b_m-\theta_{m,i}), \quad \forall i.
\end{equation}
Then, the normalized prior is given by
\begin{equation}
	q_{m,\zeta,i}
	=
	\frac{
		\tilde{q}_{m,\zeta,i}
	}{
		\sum_{i=1}^{N_{\theta_m}}
		\tilde{q}_{m,\zeta,i}
		\Delta\theta_m
	},
	\quad
	\forall i.
\end{equation}
Then, a discrete approximation of the Fisher information of the prior distribution $q_{m,\zeta}(\theta_m)$ is given by
\begin{equation}
	\mathbf{\Phi}_{\mathsf{p},\zeta}^{(m,m)}
	=
	\sum_{i=1}^{N_{\theta_m}}
	r_{m,\zeta,i}^2
	q_{m,\zeta,i}
	\Delta\theta_m,
\end{equation}
where
\begin{equation}
	r_{m,\zeta,i}
	=
	\frac{
		h_\zeta'(\theta_{m,i}-a_m)
	}{
		h_\zeta(\theta_{m,i}-a_m)
	}
	-
	\frac{
		h_\zeta'(b_m-\theta_{m,i})
	}{
		h_\zeta(b_m-\theta_{m,i})
	}.
\end{equation}

\bibliographystyle{IEEEtran} 
\bibliography{reference}

\begin{thebibliography}{10}
\providecommand{\url}[1]{#1}
\csname url@samestyle\endcsname
\providecommand{\newblock}{\relax}
\providecommand{\bibinfo}[2]{#2}
\providecommand{\BIBentrySTDinterwordspacing}{\spaceskip=0pt\relax}
\providecommand{\BIBentryALTinterwordstretchfactor}{4}
\providecommand{\BIBentryALTinterwordspacing}{\spaceskip=\fontdimen2\font plus
\BIBentryALTinterwordstretchfactor\fontdimen3\font minus
  \fontdimen4\font\relax}
\providecommand{\BIBforeignlanguage}[2]{{%
\expandafter\ifx\csname l@#1\endcsname\relax
\typeout{** WARNING: IEEEtran.bst: No hyphenation pattern has been}%
\typeout{** loaded for the language `#1'. Using the pattern for}%
\typeout{** the default language instead.}%
\else
\language=\csname l@#1\endcsname
\fi
#2}}
\providecommand{\BIBdecl}{\relax}
\BIBdecl

\bibitem{LiuISIT}
Y.~Liu and W.~Yu, ``Precoder design for {MIMO} sensing with scatterers,'' in
  \emph{IEEE Int. Symp. Inf. Theory (ISIT)}, Jun. 2026.

\bibitem{6736761}
E.~G. Larsson, O.~Edfors, F.~Tufvesson, and T.~L. Marzetta, ``Massive {MIMO}
  for next generation wireless systems,'' \emph{IEEE Commun. Mag.}, vol.~52,
  no.~2, pp. 186--195, Feb. 2014.

\bibitem{9737357}
F.~Liu, Y.~Cui, C.~Masouros, J.~Xu, T.~X. Han, Y.~C. Eldar, and S.~Buzzi,
  ``Integrated sensing and communications: Toward dual-functional wireless
  networks for {6G} and beyond,'' \emph{IEEE J. Sel. Areas Commun.}, vol.~40,
  no.~6, pp. 1728--1767, Jun. 2022.

\bibitem{1367565}
C.~Doan, S.~Emami, D.~Sobel, A.~Niknejad, and R.~Brodersen, ``Design
  considerations for 60 {GHz} {CMOS} radios,'' \emph{IEEE Commun. Mag.},
  vol.~42, no.~12, pp. 132--140, Dec. 2004.

\bibitem{4276989}
P.~Stoica, J.~Li, and Y.~Xie, ``On probing signal design for {MIMO} radar,''
  \emph{IEEE Trans. Signal Process.}, vol.~55, no.~8, pp. 4151--4161, Aug.
  2007.

\bibitem{5765721}
S.~Ahmed, J.~S. Thompson, Y.~R. Petillot, and B.~Mulgrew, ``Unconstrained
  synthesis of covariance matrix for {MIMO} radar transmit beampattern,''
  \emph{IEEE Trans. Signal Process.}, vol.~59, no.~8, pp. 3837--3849, Aug.
  2011.

\bibitem{4359542}
J.~Li, L.~Xu, P.~Stoica, K.~W. Forsythe, and D.~W. Bliss, ``Range compression
  and waveform optimization for {MIMO} radar: A {CramÉr–Rao} bound based
  study,'' \emph{IEEE Trans. Signal Process.}, vol.~56, no.~1, pp. 218--232,
  Jan. 2008.

\bibitem{10138058}
X.~Song, J.~Xu, F.~Liu, T.~X. Han, and Y.~C. Eldar, ``Intelligent reflecting
  surface enabled sensing: {Cramér-Rao} bound optimization,'' \emph{IEEE
  Trans. Signal Process.}, vol.~71, pp. 2011--2026, 2023.

\bibitem{8386661}
F.~Liu, L.~Zhou, C.~Masouros, A.~Li, W.~Luo, and A.~Petropulu, ``Toward
  dual-functional radar-communication systems: Optimal waveform design,''
  \emph{IEEE Trans. Signal Process.}, vol.~66, no.~16, pp. 4264--4279, Aug.
  2018.

\bibitem{9124713}
X.~Liu, T.~Huang, N.~Shlezinger, Y.~Liu, J.~Zhou, and Y.~C. Eldar, ``Joint
  transmit beamforming for multiuser {MIMO} communications and {MIMO} radar,''
  \emph{IEEE Trans. Signal Process.}, vol.~68, pp. 3929--3944, 2020.

\bibitem{9652071}
F.~Liu, Y.-F. Liu, A.~Li, C.~Masouros, and Y.~C. Eldar, ``{Cramér-Rao} bound
  optimization for joint radar-communication beamforming,'' \emph{IEEE Trans.
  Signal Process.}, vol.~70, pp. 240--253, 2022.

\bibitem{2509.13661}
\BIBentryALTinterwordspacing
K.~M. Attiah and W.~Yu, ``Uplink-downlink duality for beamforming in integrated
  sensing and communications,'' \emph{Accepted in IEEE J. Sel. Areas Inf.
  Theory}, 2026. [Online]. Available: \url{https://arxiv.org/abs/2509.13661}
\BIBentrySTDinterwordspacing

\bibitem{10147248}
Y.~Xiong, F.~Liu, Y.~Cui, W.~Yuan, T.~X. Han, and G.~Caire, ``On the
  fundamental tradeoff of integrated sensing and communications under
  {Gaussian} channels,'' \emph{IEEE Trans. Inf. Theory}, vol.~69, no.~9, pp.
  5723--5751, Sep. 2023.

\bibitem{10440056}
X.~Song, X.~Qin, J.~Xu, and R.~Zhang, ``{Cramér-Rao} bound minimization for
  {IRS}-enabled multiuser integrated sensing and communications,'' \emph{IEEE
  Trans. Wireless Commun.}, vol.~23, no.~8, pp. 9714--9729, Aug. 2024.

\bibitem{9454375}
S.~Buzzi, E.~Grossi, M.~Lops, and L.~Venturino, ``Radar target detection aided
  by reconfigurable intelligent surfaces,'' \emph{IEEE Signal Process. Lett.},
  vol.~28, pp. 1315--1319, 2021.

\bibitem{9732186}
------, ``Foundations of {MIMO} radar detection aided by reconfigurable
  intelligent surfaces,'' \emph{IEEE Trans. Signal Process.}, vol.~70, pp.
  1749--1763, 2022.

\bibitem{9769997}
R.~Liu, M.~Li, Y.~Liu, Q.~Wu, and Q.~Liu, ``Joint transmit waveform and passive
  beamforming design for {RIS}-aided {DFRC} systems,'' \emph{IEEE J. Sel.
  Topics Signal Process.}, vol.~16, no.~5, pp. 995--1010, Aug. 2022.

\bibitem{10364735}
R.~Liu, M.~Li, Q.~Liu, and A.~Lee~Swindlehurst, ``{SNR/CRB}-constrained joint
  beamforming and reflection designs for {RIS-ISAC} systems,'' \emph{IEEE
  Trans. Wireless Commun.}, vol.~23, no.~7, pp. 7456--7470, Jul. 2024.

\bibitem{6748061}
S.~Ahmed and M.-S. Alouini, ``Mimo-radar waveform covariance matrix for high
  {SINR} and low side-lobe levels,'' \emph{IEEE Trans. Signal Process.},
  vol.~62, no.~8, pp. 2056--2065, Apr. 2014.

\bibitem{7953658}
B.~Li and A.~P. Petropulu, ``Joint transmit designs for coexistence of {MIMO}
  wireless communications and sparse sensing radars in clutter,'' \emph{IEEE
  Trans. Aerosp. Electron. Syst.}, vol.~53, no.~6, pp. 2846--2864, Dec. 2017.

\bibitem{6649991}
G.~Cui, H.~Li, and M.~Rangaswamy, ``{MIMO} radar waveform design with constant
  modulus and similarity constraints,'' \emph{IEEE Trans. Signal Process.},
  vol.~62, no.~2, pp. 343--353, Jan. 2014.

\bibitem{8141978}
Z.~Cheng, Z.~He, B.~Liao, and M.~Fang, ``{MIMO} radar waveform design with
  {PAPR} and similarity constraints,'' \emph{IEEE Trans. Signal Process.},
  vol.~66, no.~4, pp. 968--981, Feb. 2018.

\bibitem{8892631}
Z.~Cheng, B.~Liao, S.~Shi, Z.~He, and J.~Li, ``Co-design for overlaid {MIMO}
  radar and downlink {MISO} communication systems via {Cramér–Rao} bound
  minimization,'' \emph{IEEE Trans. Signal Process.}, vol.~67, no.~24, pp.
  6227--6240, Dec. 2019.

\bibitem{Liu2506:MIMO}
Y.~Liu and W.~Yu, ``{MIMO} sensing beamforming design with low-resolution
  transceivers,'' in \emph{IEEE Int. Conf. Commun. (ICC)}, Montreal, Canada,
  Jun. 2025.

\bibitem{9266762}
Z.~Cheng, S.~Shi, L.~Tang, Z.~He, and B.~Liao, ``Waveform design for collocated
  {MIMO} radar with high-mix-low-resolution {ADCs},'' \emph{IEEE Trans. Signal
  Process.}, vol.~69, pp. 28--41, 2021.

\bibitem{9399801}
Z.~Cheng, S.~Shi, Z.~He, and B.~Liao, ``Transmit sequence design for
  dual-function radar-communication system with one-bit {DACs},'' \emph{IEEE
  Trans. Wireless Commun.}, vol.~20, no.~9, pp. 5846--5860, Sep. 2021.

\bibitem{10901673}
Y.~Wang and S.~Zhang, ``Hybrid beamforming design for integrated sensing and
  communication exploiting prior information,'' in \emph{IEEE Global Commun.
  Conf. (GLOBECOM)}, 2024, pp. 4576--4581.

\bibitem{6541985}
W.~Huleihel, J.~Tabrikian, and R.~Shavit, ``Optimal adaptive waveform design
  for cognitive {MIMO} radar,'' \emph{IEEE Trans. Signal Processing}, vol.~61,
  no.~20, pp. 5075--5089, 2013.

\bibitem{10901183}
J.~Yao and S.~Zhang, ``Optimal transmit signal design for multi-target {MIMO}
  sensing exploiting prior information,'' in \emph{IEEE Global Commun. Conf.
  (GLOBECOM)}, Cape Town, South Africa, Dec. 2024.

\bibitem{10584278}
C.~Xu and S.~Zhang, ``{MIMO} integrated sensing and communication exploiting
  prior information,'' \emph{IEEE J. Sel. Areas Commun.}, vol.~42, no.~9, pp.
  2306--2321, Sept. 2024.

\bibitem{skolnik2001radar}
M.~I. Skolnik, \emph{Introduction to Radar Systems}, 3rd~ed.\hskip 1em plus
  0.5em minus 0.4em\relax New York, NY, USA: McGraw-Hill, 2001.

\bibitem{6847111}
A.~Alkhateeb, O.~El~Ayach, G.~Leus, and R.~W. Heath, ``Channel estimation and
  hybrid precoding for millimeter wave cellular systems,'' \emph{IEEE J. Sel.
  Topics Signal Process.}, vol.~8, no.~5, pp. 831--846, Oct. 2014.

\bibitem{8750903}
Y.~Shabara, C.~E. Koksal, and E.~Ekici, ``Beam discovery using linear block
  codes for millimeter wave communication networks,'' \emph{IEEE/ACM Trans.
  Netw.}, vol.~27, no.~4, pp. 1446--1459, Aug. 2019.

\bibitem{10124207}
T.~Jiang, F.~Sohrabi, and W.~Yu, ``Active sensing for two-sided beam alignment
  and reflection design using ping-pong pilots,'' \emph{IEEE J. Sel. Areas Inf.
  Theory}, vol.~4, pp. 24--39, May 2023.

\bibitem{11114787}
Y.~Liu, K.~M. Attiah, and W.~Yu, ``{RIS}-assisted joint sensing and
  communications via fractionally constrained fractional programming,''
  \emph{IEEE Trans. Wireless Commun.}, vol.~25, pp. 1674--1689, 2026.

\bibitem{rockafellar1970convex}
R.~T. Rockafellar, \emph{Convex Analysis}.\hskip 1em plus 0.5em minus
  0.4em\relax Princeton, NJ, USA: Princeton University Press, 1970.

\bibitem{10605808}
K.~Shen, Z.~Zhao, Y.~Chen, Z.~Zhang, and H.~Victor~Cheng, ``Accelerating
  quadratic transform and {WMMSE},'' \emph{IEEE J. Sel. Areas Commun.},
  vol.~42, no.~11, pp. 3110--3124, 2024.

\bibitem{mosek2019}
{MOSEK ApS}, ``{MOSEK optimization toolbox for MATLAB},'' \emph{{U}ser’s
  Guide Reference Manual}, vol.~4, no.~1, 2019.

\bibitem{8314727}
K.~Shen and W.~Yu, ``Fractional programming for communication systems—{P}art
  {I}: Power control and beamforming,'' \emph{IEEE Trans. Signal Process.},
  vol.~66, no.~10, May 2018.

\bibitem{7547360}
Y.~Sun, P.~Babu, and D.~P. Palomar, ``Majorization-minimization algorithms in
  signal processing, communications, and machine learning,'' \emph{IEEE Trans.
  Signal Process.}, vol.~65, no.~3, pp. 794--816, Feb. 2017.

\bibitem{10584423}
O.~Aharon and J.~Tabrikian, ``Asymptotically tight {Bayesian} {Cramér-Rao}
  bound,'' \emph{IEEE Trans. Signal Process.}, vol.~72, pp. 3333--3346, 2024.

\end{thebibliography}

\vfill

\end{document}